\documentclass{IEEEtran}
\usepackage{algorithmic}
\usepackage{graphicx}
\usepackage{amsmath,amsfonts,amssymb,amscd}
\usepackage{enumerate}
\usepackage{subcaption}
\usepackage{xcolor}
\usepackage[english]{babel}
\usepackage{hyperref}
\usepackage{textcomp}

\allowdisplaybreaks[4]
\newtheorem{proposition}{\bfseries Proposition}
\newtheorem{example}{\bfseries Example}
\newtheorem{assumption}{\it Assumption}
\newtheorem{theorem}{\bfseries Theorem}
\newtheorem{lemma}{\bfseries Lemma}
\newtheorem{remark}{\bfseries Remark}
\newtheorem{problem}{\bfseries Problem}
\newtheorem{proof}{Proof}

\newcommand{\R}{\mathbb{R}}
\newcommand{\C}{\mathcal{C}}
\newcommand{\fot}{\frac{1}{2}}
\newcommand{\ka}{\kappa}
\newcommand{\sg}{\sigma}
\newcommand{\ta}{\theta}
\newcommand{\om}{\omega}
\newcommand{\be}{\beta}
\newcommand{\pa}{\partial}
\newcommand{\di}{{\rm d}}
\newcommand{\la}{\lambda}
\newcommand{\ga}{\gamma}
\newcommand{\ep}{\epsilon}

\makeatletter
\renewcommand\normalsize{%
\@setfontsize\normalsize\@xpt\@xiipt
\abovedisplayskip 2.297\p@ \@plus2\p@ \@minus1\p@
\abovedisplayshortskip \z@ \@plus2\p@
\belowdisplayshortskip 2.297\p@ \@plus2\p@ \@minus1\p@
\belowdisplayskip \abovedisplayskip
\let\@listi\@listI}
\makeatother

\usepackage{xcolor}
\newif\ifdraft

\def\BibTeX{{\rm B\kern-.05em{\sc i\kern-.025em b}\kern-.08em
    T\kern-.1667em\lower.7ex\hbox{E}\kern-.125emX}}
\begin{document}
\title{Immersion and Invariance-based Disturbance Observer and Its Application to Safe Control}

\author{Yujie Wang  and Xiangru Xu, \IEEEmembership{Member, IEEE}
\thanks{Yujie Wang and Xiangru Xu are with the Department of Mechanical Engineering, University of Wisconsin-Madison,
        Madison, WI 53706, USA. (e-mail: 
yujie.wang@wisc.edu;  xiangru.xu@wisc.edu).}
}

\maketitle

\begin{abstract}
When the disturbance input matrix is nonlinear, existing disturbance observer design methods rely on the solvability of a partial differential equation or the existence of an output function with a uniformly well-defined
disturbance relative degree, which can pose significant limitations. This note introduces a systematic approach for designing an Immersion and Invariance-based Disturbance Observer (IIDOB) that circumvents these strong assumptions. The proposed IIDOB ensures the disturbance estimation error is globally uniformly ultimately bounded by approximately solving a partial differential equation while compensating for the approximation error. Furthermore, by integrating IIDOB into the framework of control barrier functions, a filter-based safe control design method for control-affine systems with disturbances is established where the filter is used to generate an alternative disturbance estimation signal with a known derivative. Sufficient conditions are established to guarantee the safety of the disturbed systems. Simulation results demonstrate the effectiveness of the proposed method. 
\end{abstract}


\section{Introduction}
\label{sec:introduction}

Designing feedback controllers that guarantee the safety specification of a system has attracted significant attention in the past decades \cite{aswani2013provably,prajna2007framework,brunke2022safe,hewing2020learning}. Inspired by automotive safety applications, \cite{ames2014control,ames2016control,Xu2015ADHS} proposed reciprocal and zeroing Control Barrier Functions  (CBFs) that generalize previous barrier conditions to only require a single sub-level set to be controlled invariant. By including the CBF
condition in a convex Quadratic Program (QP), a CBF-QP-based controller is generated in real time and acts as a safety filter that modifies potentially unsafe control inputs in a minimally invasive fashion. Various robust CBF approaches have been
proposed for systems with model uncertainties and external
disturbances \cite{jankovic2018robust,nguyen2021robust,garg2021robust,buch2021robust}; however, most of these robust CBF methods consider the worst-case of disturbances, resulting in overly conservative control behaviors.

To reduce the adverse effects of disturbances/uncertainties on system performance, several works integrating disturbance/uncertainty estimation and compensation techniques into the CBF-QP framework have been proposed recently \cite{wang2022disturbance,dacs2022robust,dacs2023robust,alan2022disturbance}. In our previous work \cite{wang2022disturbance}, the Disturbance Observer (DOB) presented in \cite{chen2004disturbance} was incorporated into the CBF-QP framework for the first time. Compared with other robust control schemes, DOB-based control has two main advantages: (i) the DOB can be designed independently and added to a baseline controller to improve its robustness and disturbance attenuation capability; (ii) in the presence of disturbances/uncertainties, the nominal performance of the baseline controller can be recovered by the DOB-based controller \cite{chen2015disturbance,sariyildiz2019disturbance,li2014disturbance}.   

\begin{figure}[!t]
\centering
  \includegraphics[width=0.5\textwidth]{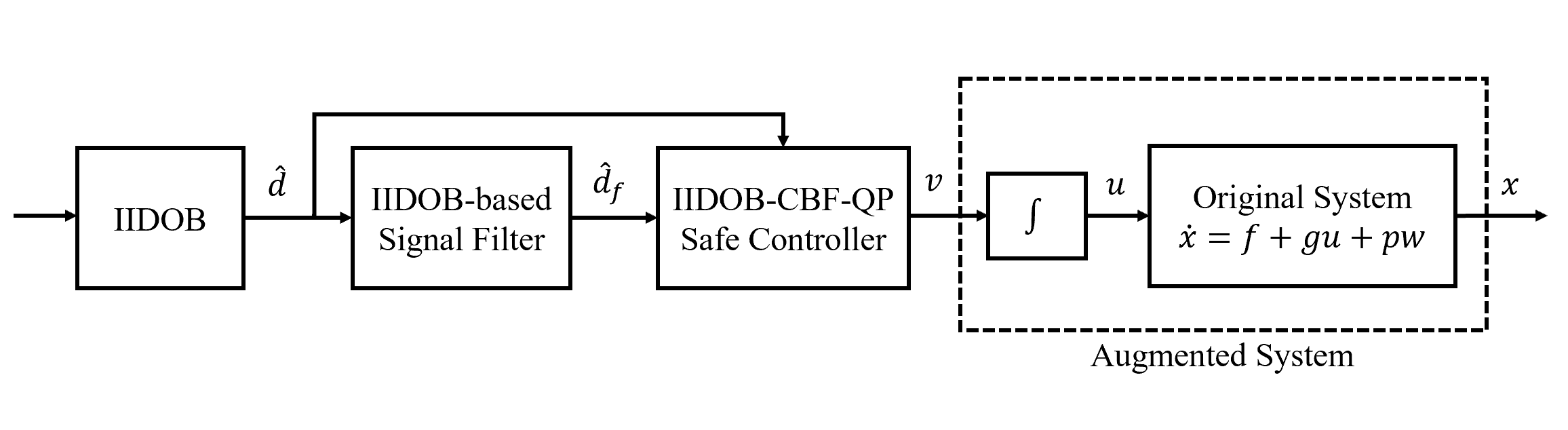}\vskip -3mm
\caption{Configuration of the proposed IIDOB-CBF-QP method that consists of three components: (i) an IIDOB used for disturbance estimation, (ii) a  filter that can generate an alternative disturbance estimation signal with a known derivative, and (iii) an IIDOB-CBF-QP-based safe controller that can ensure safety of the closed-loop system.}
\vskip -1mm
    \label{fig0}
\end{figure}

Nevertheless, the design of DOBs is non-trivial and highly
problem-specific. Specifically, designing a DOB requires the existence of two functions that can ensure the asymptotic stability of the error dynamics and the satisfaction of a partial differential equation (PDE) simultaneously (more details will be given in the next section) \cite{chen2015disturbance}. Fulfilling these two requirements is challenging, and 
existing methods rely on relatively strong assumptions, e.g., the disturbance relative degree is uniformly well-defined \cite{chen2004disturbance,chen1999nonlinear}. A systematic and computationally feasible method for constructing DOBs for generic nonlinear control-affine systems is still lacking.

The contribution of this note is twofold: 
(i) Inspired by the Immersion and Invariance (I\&I) technique \cite{astolfi2003immersion,astolfi2008nonlinear,astolfi2009globally}, we propose a systematic approach for designing I\&I-based Disturbance Observer (IIDOB) for general nonlinear control-affine systems without imposing the strong assumptions adopted by existing DOB design methods, such as the solvability of a PDE or the existence of an output function with a uniformly well-defined disturbance relative degree. 
By approximately solving the PDE and compensating for the approximation error, the proposed IIDOB ensures that the disturbance estimation error is globally Uniformly Ultimately Bounded (UUB).
(ii) We propose a filter-based IIDOB-CBF-QP safe control design approach for control-affine systems with disturbances (see Fig. \ref{fig0}). We design a filter to obtain an alternative disturbance estimation signal with a known derivative and provide sufficient conditions that ensure the safety of the disturbed system. 
The remainder of this note is organized as follows: the background and the problem statement are provided in Section \ref{sec:preliminary}; the proposed IIDOB is presented in Section \ref{sec:dob}; the IIDOB-CBF-QP-based safe control strategy is provided in Section \ref{sec:dobcbf}; numerical simulation results are provided 
in Section \ref{sec:simulation}; and finally, the conclusion is drawn in Section \ref{sec:conclusion}.

\subsubsection*{\textbf{Notation}}
For a given positive integer $n$, denote $[n]=\{1,2,\cdots,n\}$. For a column vector $x\in\R^{n}$ or a row vector $x\in\R^{1\times n}$, denote $x_i$ as the $i$-th entry of $x$ and $\|x\|$ as its 2-norm. Denote $I_{n}$ as an identity matrix with dimension $n\times n$. For a given matrix $A\in\R^{n\times m}$, $A_{ij}$ is the $(i,j)$-th entry of $A$, $A_j$ is the $j$-th column of $A$, and $\|A\|$ is its Frobenius norm. Denote ${\rm diag}[a_1,a_2,\cdots,a_n]\in\R^{n\times n}$ as a diagonal matrix with diagonal entries $a_1,a_2,\cdots,a_n\in\R$.
The gradient $\frac{\pa h}{\pa x} \in \R^{n\times 1}$ is considered as a row vector, where $x\in\R^n$ and $h:\R^n\to\R$ is a function with respect to $x$. For a function $f:\R^n \to \R^m$ with respect to $x\in\R^n$, $\frac{\pa f}{\pa x}$ denotes the Jacobian matrix whose $(i,j)$-th entry is $\frac{\pa f_i}{\pa x_j}$.

\section{Background and Problem Statement}
\label{sec:preliminary}

\subsection{Background}
\label{sec:motivation}

Consider a control-affine system  $\dot x = f(x)+g(x) u$,
where $x\in\R^n$ is the state, $u\in\R^m$ 
is the control input, and $f: \R^n\to\R^n$ and $g: \R^n\to\R^{n\times m}$ are known and locally Lipchitz continuous
function functions. Define a safe set $\mathcal{C}$ as  
\begin{equation}\label{setc}
    \mathcal{C} = \{ x \in \R^n : h(x) \geq 0\},
\end{equation}
where $h:\R^n\to\R$ is a sufficiently smooth function. The function $h$ is called a CBF of (input) relative degree 1
if  
$
 \sup_{u\in\R^m}  \left[ L_f h + L_{g} h u + \gamma h\right] \geq 0
$ holds
for all $x\in\R^n$, where $\gamma>0$ is a given positive constant, and $L_fh=\frac{\pa h}{\pa x}f$ and $L_gh=\frac{\pa h}{\pa x}g$ are Lie derivatives  \cite{Xu2015ADHS}. It was proven in  \cite{Xu2015ADHS} that any Lipschitz continuous controller $u\in \{ {\mathfrak u} : L_f h + L_{g} h {\mathfrak u} + \gamma h \geq 0\}$ will ensure the safety of the closed-loop system, i.e., the forward invariance of $\C$.

Now consider the following control-affine system with disturbances:
\begin{equation}
    \dot x = f(x)+g(x) u+p(x) w(t)\label{controlaffine}
\end{equation}
where $x\in\R^n$ is the state, $u\in\R^m$ is the control input, $w:\R_{\geq 0}\to\R^l$ is the disturbance, and $f:\R^n\to\R^n$, $g:\R^n\to\R^{n\times m}$, and $p:\R^n\to\R^{n\times l}$ are known functions. Provided that the disturbance $w$ is bounded, robust CBF-based methods can be adopted to ensure the safety of system \eqref{controlaffine} \cite{jankovic2018robust,nguyen2021robust}. In existing robust CBF-based methods, safety is achieved by sacrificing the nominal performance, as the worst-case of the disturbances is considered in the safe controller design. Therefore, trajectories of the closed-loop system will stay in a shrunk subset of the original safe set $\C$, implying that the performance of these controllers is conservative.

DOB is one of the most effective tools for estimating and compensating disturbances/uncertainties in nonlinear control design, and has been extensively applied to numerous systems \cite{chen2015disturbance,sariyildiz2019disturbance,li2014disturbance}. 
Our previous work \cite{wang2022disturbance} integrated DOBs into the CBF-QP framework, and proposed a DOB-CBF-QP controller with safety guarantees. However, designing DOBs for control-affine system \eqref{controlaffine} is non-trivial.

Suppose that $w$ is slowly time-varying, that is,  $\dot w\approx 0$ for any $ t\geq 0$. 
As shown in \cite{chen2015disturbance,chen1999nonlinear}, the DOB for system \eqref{controlaffine} has the following structure:
\begin{IEEEeqnarray}{rCl}
\IEEEyesnumber\label{chendoblaw}
\IEEEyessubnumber
    \hat w &=& z + q(x),\\
\IEEEyessubnumber
    \dot z &=& -l(x)p(x) z-l(x)
    [f(x)+g(x)u+p(x)q(x)],
\end{IEEEeqnarray}
where $\hat w$ is the disturbance estimation,
$z\in\R^l$ is the internal state of the DOB.
The function $l(x)$, known as the DOB gain, and the function $q(x)$ should be designed such that \cite{chen2015disturbance}
\begin{equation}\label{eqPDE}
    \frac{\pa q(x)}{\pa x}=l(x),
\end{equation}
and the error dynamics is globally asymptotically stable for any $x\in\R^n$:
\begin{equation}\label{eqerror}
    \dot e_w+l(x)p(x)e_w=0,
\end{equation}
where $e_w=w-\hat w$ is the disturbance estimation error. 

Designing $l$ and $q$ is a challenging and highly case-specific task in general \cite{chen2015disturbance}. Several methods have been proposed based on relatively strong assumptions. If $p$ has full column rank for any $x\in\R^n$, then one can select $q$ by solving the PDE $\frac{\pa q}{\pa x}=p^\dagger$, where $p^\dagger$ denotes the left inverse of $p$ \cite{mohammadi2017nonlinear}; however, when $n>1$, this PDE is generally unsolvable and even when solvable, its closed-form solution is hard to obtain. If the disturbance relative degree is uniformly well-defined with respect to an output function $s(x)$, an approach for designing $q$ is proposed in  \cite{chen2004disturbance,chen1999nonlinear}; however selecting such a function $s$ is challenging and its existence is not guaranteed (e.g., there may exist $x^*\in\R^n$ such that $p(x^*)=0$). A practically useful approach 
involves treating $pw$ as the total disturbance and assuming that $\frac{\di}{\di t}(pw)$ is bounded \cite{an2016disturbance,mohammadi2011disturbance}; however, this assumption is rather restrictive because $\frac{\di}{\di t}(pw)$ explicitly relies on $u$ and $x$. 
As will be shown in Section \ref{sec:dob}, we will provide a systematic approach for designing DOBs that avoids the issues of the aforementioned methods.

\subsection{Problem Statement}
Consider system \eqref{controlaffine} and the safe set defined in \eqref{setc}, 
where $h$ is a sufficiently smooth function. Recall that $g_j$ denotes the $j$-th column of $g$ for $j\in[m]$, and $p_i$ denotes the $i$-th column of $p$ for $i\in[l]$.  System \eqref{controlaffine} is said to have a vector Input Relative Degree (IRD) $\mathcal{I}=(\sigma_1,\sigma_2,\cdots,\sigma_m)$ at a given point $x_0\in\R^n$ if $L_{g_j}L_f^kh(x)=0$ for any $k\in[\sigma_j-2]$, $j\in[m]$, and for all $x$ in a neighborhood of $x_0$, and $L_{g_j}L_f^{\sg_j-1}h(x_0)\neq 0$ holds for any $j\in[m]$ \cite[Remark 5.1.1]{isidori1985nonlinear}. Similarly, system \eqref{controlaffine} is said to have a vector Disturbance Relative Degree (DRD) $\mathcal{D}=(\nu_1,\nu_2,\cdots,\nu_l)$ at a given point $x_0\in\R^n$ if $L_{p_j}L_f^kh(x)=0$ for any $k\in[\nu_j-2]$, $j\in[l]$, and for all $x$ in a neighborhood of $x_0$, and $L_{p_j}L_f^{\nu_j-1}h(x_0)\neq 0$ holds for any $j\in[l]$ \cite{yang2013nonlinear}. Note that because system \eqref{controlaffine} is multiple-input-single-output with $h$ as the output, the definitions of vector IRD and vector DRD above are slight modifications of those given in \cite{isidori1985nonlinear,yang2013nonlinear}. 

In this note, with a slight abuse of notation, we will call $r_I=\min \mathcal{I}$ and $r_D=\min \mathcal{D}$ as the minimum IRD and the minimum DRD of system \eqref{controlaffine} with respect to function $h$ at a given point $x_0\in\R^n$, respectively; that is, $r_I$ (or $r_D$) denotes the number of times $h$ has to be differentiated to have at least one component of $u$ (or $w$) explicitly appearing.

Next, a standard assumption for DOB design is given.
\begin{assumption}\label{assmp:d}
The disturbance $w$ and its derivative are bounded as $\|w\|\leq \om_0$ and $\|\dot w\|\leq \om_1$, where $\om_0$ and $\om_1$ are positive constants not necessarily known in DOB design.
\end{assumption}

The first problem investigated in this note is to design a disturbance estimation law to estimate the total disturbance 
\begin{equation}
    d(x,t)= p(x)w(t).\label{totaldisturb}
\end{equation}

\begin{problem}\label{prob1}
Consider system \eqref{controlaffine} with $f,g\in C^1$ and $p\in C^2$ and suppose Assumption \ref{assmp:d} holds. Design a DOB-based estimation law to estimate the total disturbance $d$ online.
\end{problem}

Using the DOB-based estimation of the total disturbance, the second problem investigated in this note is to design a feedback control law such that system \eqref{controlaffine} is safe.

\begin{problem}\label{prob2}
Consider system \eqref{controlaffine} with $f,g\in C^1$ and $p\in C^2$ and the safe set $\mathcal{C}$ defined in \eqref{setc}. Suppose that Assumption \ref{assmp:d} holds and $r_I=r_D$ for system \eqref{controlaffine} with respect to $h$. 
Given the DOB developed via solving Problem \ref{prob1}, design a feedback control law such that system \eqref{controlaffine} is safe, 
i.e., $h(x(t))\geq 0$ for any $t>0$ provided $h(x(0))> 0$.
\end{problem}

\begin{remark}
In Problem \ref{prob2}, if $r_I<r_D$, the disturbance can be directly decoupled from the system via state feedback control \cite{yang2013static}. The case $r_I>r_D$ will be explored in our future work. 
Note that we don't assume the minimum DRD of system \eqref{controlaffine} is uniformly well-defined as in \cite{chen2004disturbance,chen1999nonlinear}, i.e., there may exist $x_0\in\R^n$ such that $L_{p_j}L_f^{\nu_j-1}h(x_0)=0$ for any $j\in[l]$.
\end{remark}

\begin{remark}
In this note, we aim to estimate the total disturbance $d$ rather than the disturbance $w$ for two main reasons: first, since no assumption is imposed on $p$ except for $p\in C^2$, the disturbance $w$ may not be uniquely determined in general;  second, Problem 2 can be solved by using the information of $d$ only. 
\end{remark}

\section{IIDOB Design}
\label{sec:dob}
Inspired by the I\&I technique \cite{astolfi2003immersion,astolfi2008nonlinear,astolfi2009globally},
we propose an IIDOB design approach to solve Problem \ref{prob1} in this section. 

First, we augment system \eqref{controlaffine} with an additional integrator: 
\begin{IEEEeqnarray}{rCl}
\IEEEyesnumber\label{eqnsys}
\IEEEyessubnumber\label{eqnsys1}
\dot x&=& f(x)+g(x)u+p(x)w(t),\\
\IEEEyessubnumber\label{eqnsys2}
\dot u&=&v,
\end{IEEEeqnarray}
where $v$ denotes the auxiliary control input to be designed, and $u$ is considered as a state variable of the augmented system. The relationship between system \eqref{controlaffine} and the augmented system \eqref{eqnsys} is illustrated in Fig. \ref{fig0}. As will be shown in this section and Section \ref{sec:dobcbf}, the auxiliary control input $v$ will be used in the design of IIDOB, and it will be generated from solving the IIDOB-CBF-QP. The control input $u$ for the original system \eqref{controlaffine} will be obtained through integrating $v$.

Define a time-varying set $\mathcal{M}(t)=\{(x,\hat x,u)\in\R^n\times\R^n\times\R^m:\xi(t)+\beta(\hat x,x,u)-p(x)w(t)=0\}$,
where $\hat x$ denotes the state estimation and $\xi$, $\beta$ are known functions that will be all specified later. Define 
\begin{equation}
\hat d= \xi+\beta
\end{equation}
as the estimated total disturbance, and the disturbance estimation error as 
\begin{equation}
e_d=\hat d-d.\label{ed}
\end{equation}
It is clear that if the system trajectories are restricted to  $\mathcal{M}(t)$, the disturbance estimation is accurate.
We also define 
\begin{equation}
    z=\frac{\xi +\beta -d}{r},\label{eqz}
\end{equation}
where $r$ is the scaling factor governed by an adaptive law yet to be designed. It is clear that $e_d=rz$. Our IIDOB design will render  $e_d$ globally UUB \cite[Definition 4.6]{khalil2002nonlinear} by guaranteeing that $z$ is globally UUB and $r$ remains bounded. 
Note that $\dot z$, the time derivative of $z$, can be expressed as
\begin{IEEEeqnarray}{rCl}
    \hspace{-2mm}\dot z&=&-\frac{\dot r}{r}z+\frac{1}{r}\bigg(
    \dot\xi+\frac{\pa\beta}{\pa x}(f+gu+pw)+\frac{\pa\beta}{\pa u}v+\frac{\pa\beta}{\pa \hat x}\dot{\hat{x}}\nonumber\\
    \hspace{-2mm}&&-p\dot w- \sum_{i=1}^l\frac{\pa p_{i}}{\pa x}(f+gu+pw)w_i
    \bigg),\label{dotz0}
\end{IEEEeqnarray}
where $p_{i}$ denotes the $i$-th column of $p$, $i\in[l]$.
Define
\begin{IEEEeqnarray}{rCl}
\hspace{-5mm}\psi(x,u)\!=\!\frac{\eta}{2}\!\left[\!\|p\|^2\!\!+\!\! \sum_{i=1}^l\! \left(
    \left\|\frac{\pa p_{i}}{\pa x}(f\!+\!gu)\right\|^2\!\!\!\!+\!\left\|\frac{\pa p_{i}}{\pa x}p\right\|^2\right)\!\right]\!\!+\!\ga\label{psi}
\end{IEEEeqnarray}
where $\ga,\eta>0$ are tuning parameters and $\ga$ denotes the observer gain. 

If $\delta(x,u)\in\R^n$ is a solution to the following PDE:
\begin{equation}
    \frac{\pa \delta}{\pa x}=\psi I_{n},\label{pde3}
\end{equation}
then the DOB design becomes straightforward by following \cite{astolfi2003immersion}. Specifically, one can design the total disturbance estimation as $\hat d = \xi+\delta$, where $\xi$ is governed by $\dot\xi=-\frac{\pa\delta}{\pa x}(f+gu+\xi+\delta)-\frac{\pa\delta}{\pa u}v$. Invoking \eqref{eqnsys} and \eqref{ed}, one can see $\dot e_d=-\psi e_d-p\dot w-\sum_{i=1}^l\frac{\pa p_{i}}{\pa x}(f+gu+pw)w_i$. By selecting a candidate Lyapunov  function $V=\fot e_d^\top e_d$, one can easily verify that $\dot V=e_d^\top \left(-\psi e_d-p\dot w - \sum_{i=1}^l \frac{\pa p_i}{\pa x}(f+gu+pw)w_i \right)\leq -\psi \|e_d\|^2+\om_1\|p\|\|e_d\|+\sum_{i=1}^l\om_0\left\|\frac{\pa p_{i}}{\pa x}(f+gu)\right\|\|e_d\|+\sum_{i=1}^l\om_0^2\left\|\frac{\pa p_{i}}{\pa x}p\right\|\|e_d\|\leq -\ga \|e_d\|^2+\frac{1}{2\eta}(\om_1^2+l\om_0^2+l\om_0^4)$, which indicates that $e_d$ is globally UUB.
However, when $n>1$, solving \eqref{pde3} is extremely challenging in principle, and even a solution to \eqref{pde3} may not exist \cite{astolfi2009globally}. To tackle this issue, we will follow \cite{astolfi2009globally} to first ``approximately solve"  \eqref{pde3} and then use $\dot r$ to compensate for the approximation error.

Recall that $f,g\in C^1$ and $p\in C^2$. Then, $\psi \in C^1$, and it is easy to verify that there exist continuous functions   $\delta_{ij}:\R^n\times\R^n\times\R^m\to\R$, $i,j\in[n]$, such that \cite{sonneveldt2010immersion,karagiannis2008observer}: 
\begin{IEEEeqnarray}{rCl}
    &&\psi(\hat x_{1}, \cdots,\hat x_{i-1}, x_{i}, \hat x_{i+1}, \cdots, \hat x_{n},u)-\psi(x,u)\nonumber\\
    &&=-\sum_{j=1}^n\delta_{ij}(x,\hat x,u)e_j,\label{hi}
\end{IEEEeqnarray}
where 
\begin{equation}
    e= \hat x-x\label{definitione}
\end{equation}
and $e_j$ denotes the $j$-th entry of $e$. 
The following theorem shows that our IIDOB design ensures the disturbance estimation error $e_d$ is globally UUB.

\begin{theorem}\label{theorem:dob}
Consider system \eqref{eqnsys} where $f,g\in C^1$ and $p\in C^2$, and suppose
Assumption \ref{assmp:d} holds. If the disturbance estimation law $\hat d$ is designed as:
\begin{IEEEeqnarray}{rCl}
\IEEEyesnumber\label{dob}
\IEEEyessubnumber\label{dob:hatd}
\hat d&=& \xi+\beta,\\
 \IEEEyessubnumber\label{designk}
 \Lambda&= &(k_1\!+\!k_2r^2)I_{n}\! + \!\frac{c r^2}{2}{\rm diag}[\|\Delta_1\|^2,\cdots,\|\Delta_n\|^2],\\
\IEEEyessubnumber\label{dob:beta}
\beta&=&\begin{bmatrix}
    \int_{0}^{x_1} \psi(\tau,\hat x_2,\hat x_3,\cdots,\hat x_n,u){\rm d}\tau\\
    \int_{0}^{x_2} \psi(\hat x_1,\tau,\hat x_3,\cdots,\hat x_n,u){\rm d}\tau\\
    \vdots\\
    \int_{0}^{x_n} \psi(\hat x_1,\hat x_2,\cdots,\hat x_{n-1}, \tau,u){\rm d}\tau
    \end{bmatrix},\\
\IEEEyessubnumber\label{dob:hatx}
\dot{\hat{x}}&=&f+gu+\hat d-\Lambda e,\\
\IEEEyessubnumber\label{dob:xi}
\dot\xi&=&-\frac{\pa\beta}{\pa x}(f+gu+\hat d)-\frac{\pa\beta}{\pa u}v-\frac{\pa\beta}{\pa \hat x}\dot{\hat{x}},\\
\IEEEyessubnumber\label{dob:r}
\dot r&=&-\theta(r\!-\!1)\!+\!\frac{c r}{2}\!\sum_{j=1}^n\!e_j^2\|\Delta_j\|^2,\ r(0)>1,
\end{IEEEeqnarray}
where $\psi$ is defined in \eqref{psi}, $e$ is defined in \eqref{definitione}, $\ga,c,\ta>0$ are positive constants satisfying $\ga>\frac{n}{2c}+\ta$, $\Delta_j=\text{diag}[\delta_{1j}, \delta_{2j},\cdots, \delta_{nj}]\in\R^{n\times n}$ with $\delta_{ij}$ defined in \eqref{hi} for $i,j\in[n]$, 
and $k_1,k_2>0$ are positive constants satisfying $k_2>\frac{1}{4\ga-2n/c-4\ta}$,
then $e_d$ is globally UUB.
\end{theorem}
\begin{proof}
Recall that $e_d=rz$. To prove that $e_d$ is globally UUB, we will first show $z$ is globally UUB, and then $r$ is bounded. 

Substituting \eqref{dob:xi} into \eqref{dotz0} yields
\begin{IEEEeqnarray}{rCl}
    \hspace{-6mm}\dot z=-\frac{\dot r}{r}z-\frac{\pa\beta}{\pa x}z-\frac{1}{r}\!\left(p\dot w+ \sum_{i=1}^l\!\frac{\pa p_{i}}{\pa x}(f\!+\!gu\!+\!pw)w_i\right)\!.\label{dotz1n}
\end{IEEEeqnarray}
Recall that $\psi\in C^1$. According to the fundamental theorem of calculus, 
one can see that 
\begin{IEEEeqnarray}{rCl}
    \frac{\pa \beta}{\pa x}&=&
    {\rm diag}[\psi(x_1,\!\hat x_2,\!\cdots,\!\hat x_n,\!u),\psi(\hat x_1,\!x_2,\!\hat x_3,\!\cdots,\!\hat x_n,\!u),\nonumber\\
    &&\cdots,\psi(\hat x_1,\!\hat x_2,\!\cdots,\!\hat x_{n-1},\!x_n,\!u)].\label{partialbetahat}
\end{IEEEeqnarray}
Define $e_\psi=\left\|\psi(x,u)I_{n}-\frac{\pa\beta}{\pa x}\right\|$ as the ``approximation error" induced by approximately solving \eqref{pde3} using $\be$ designed in \eqref{dob:beta}.
Intuitively, from \eqref{partialbetahat} one can see that if $\hat x$ is very close to $x$, $e_\psi$ would be negligible. Note that the influence of $e_\psi$ will be eliminated by $\dot r$ as shown in the following analysis.

Then, substituting \eqref{hi} into \eqref{partialbetahat} yields
\begin{IEEEeqnarray}{rCl}
    \frac{\pa\beta}{\pa x} =\psi(x,u)I_{n} -\sum_{j=1}^n\Delta_j e_j,\label{pabetaapprox}
\end{IEEEeqnarray}
and substituting \eqref{psi} and \eqref{pabetaapprox} into \eqref{dotz1n} yields
\begin{IEEEeqnarray}{rCl}
   \dot z&=&\!-\frac{\dot r}{r}z\!-\!\ga z\!- \!\frac{\eta}{2}\!\left[\sum_{i=1}^l\!\left(
    \left\|\!\frac{\pa p_{i}}{\pa x}(f\!+\!gu)\!\right\|^2\!\!\!+\left\|\frac{\pa p_{i}}{\pa x}p\right\|^2\right)\!\!+\!\|p\|^2\!\right]\!z\nonumber\\
    &&+\sum_{j=1}^n\!\Delta_je_j z\!+\!\frac{1}{r}\left(\!-p\dot w- \!\sum_{i=1}^l\frac{\pa p_{i}}{\pa x}(f\!+\!gu\!+\!pw)w_i\!\right)\!\!.\label{dotz2}
\end{IEEEeqnarray}
From \eqref{dob:r} one can easily verify that $r\geq 1$ for any $t>0$ because the set $\{r:r\geq 1\}$ is invariant. Substituting  \eqref{dob:r} into \eqref{dotz2} gives
\begin{IEEEeqnarray}{rCl}
    \dot z&=& \theta\frac{r-1}{r}z\!-\!\frac{c}{2}\sum_{j=1}^ne_j^2\|\Delta_j\|^2z\!-\! \bigg[\frac{\eta}{2}\sum_{i=1}^l\bigg(
    \left\|\frac{\pa p_{i}}{\pa x}(f+gu)\right\|^2\nonumber\\
   &&+\left\|\frac{\pa p_{i}}{\pa x}p\right\|^2\bigg)+\frac{\eta}{2}\|p\|^2-\sum_{j=1}^n\Delta_je_j\bigg]z+\frac{1}{r}\bigg(-p\dot w\nonumber\\
   &&- \sum_{i=1}^l\frac{\pa p_{i}}{\pa x}(f+gu)w_i-\sum_{i=1}^l\frac{\pa p_{i}}{\pa x}pw w_i
    \bigg)-\ga z.\label{dotz3}
\end{IEEEeqnarray}
Meanwhile, subtracting \eqref{eqnsys} from \eqref{dob:hatx} yields
\begin{IEEEeqnarray}{rCl}
    \hspace{-6mm}\dot e=\hat d-d-\Lambda e=rz-\Lambda e.\label{dote}
\end{IEEEeqnarray}
Next, we prove that $z$ is globally UUB. Define a candidate Lyapunov function as $V=\frac{1}{2}z^\top z$, whose time derivative is
\begingroup
\allowdisplaybreaks
\begin{IEEEeqnarray}{rCl}
    \dot V&\stackrel{\eqref{dotz3}}{=}&  \theta\frac{r-1}{r}\|z\|^2-\frac{c}{2}\sum_{j=1}^ne_j^2\|\Delta_j\|^2\|z\|^2+z^\top\sum_{j=1}^n\Delta_je_jz \nonumber\\
    &&-\frac{\eta}{2}\left[
   \sum_{i=1}^l\left(
    \left\|\frac{\pa p_{i}}{\pa x}(f\!+\!gu)\right\|^2\!\!\!+\!\left\|\frac{\pa p_{i}}{\pa x}p\right\|^2\right)\!+\!\|p\|^2
   \right]\|z\|^2\nonumber\\
   &&+\frac{z^\top}{r}\left(-p\dot w- \sum_{i=1}^l\frac{\pa p_{i}}{\pa x}(f+gu+pw)w_i\right)-\ga\|z\|^2\nonumber\\
   &\leq& \theta\frac{r\!-\!1}{r}\|z\|^2\!-\!\frac{c}{2}\!\sum_{j=1}^n \! e_j^2\|\Delta_j\|^2\|z\|^2\!+\!\sum_{j=1}^n\|\Delta_j\| |e_j|\|z\|^2 \nonumber\\
    &&-\frac{\eta}{2}\left[
   \sum_{i=1}^l\!\left(
    \left\|\frac{\pa p_{i}}{\pa x}(f\!+\!gu)\right\|^2\!\!\!+\!\!\left\|\frac{\pa p_{i}}{\pa x}p\right\|^2\right)\!+\!\|p\|^2
   \right]\!\!\|z\|^2\nonumber\\
   &&-\ga\|z\|^2+\frac{\|z\|}{r}\bigg(\|p\|\om_1+ \sum_{i=1}^l\!\left\|\frac{\pa p_{i}}{\pa x}(f+gu)\right\|\om_0\nonumber\\
   &&+\sum_{i=1}^l\left\|\frac{\pa p_{i}}{\pa x}p\right\|\om_0^2\bigg)\nonumber\\
    &\leq&-\left(\ga-\ta-\frac{n}{2c}\right)\|z\|^2+\frac{1}{2\eta r^2}(\om_1^2+l\om_0^2+l\om_0^4)\nonumber\\
   &\leq&-\kappa \|z\|^2+\omega,\label{dotv0}
\end{IEEEeqnarray}
\endgroup
where 
\begin{IEEEeqnarray}{rCl}
\IEEEyesnumber\label{eqkappaomega}
\IEEEyessubnumber\label{eqkappa}
\kappa&=&\ga-\frac{n}{2c}-\ta>0,\\
\IEEEyessubnumber\label{eqomega}
\omega&=&\frac{1}{2\eta}(\om_1^2+l\om_0^2+l\om_0^4)>0,    
\end{IEEEeqnarray} 
the first and second inequality arise from Cauchy-Schwarz inequality, and the last inequality comes from the fact that $r\geq 1$. 
Therefore, one can see that 
\begin{align}\label{ineqz}
\|z\|\leq \sqrt{\|z(0)\|^2e^{-2\kappa t}+\frac{\omega}{\kappa}}=\varrho_z(t),    
\end{align}
which indicates that $z$ is UUB. Note that selecting a larger $\kappa$ will result in a smaller final bound of $\|z\|$. However, the convergence of $\|z\|$ does not imply the convergence of $e_d$ unless $r$ is bounded. To show the boundedness of $r$, we construct an augmented candidate Lyapunov function $W$ as $W=V+\fot e^\top e+\fot r^2$,
whose time derivative satisfies
\begin{IEEEeqnarray}{rCl}
    \dot W\!\!&\stackrel{\eqref{dote}}{\leq}&\!\! -\kappa\|z\|^2+\omega
    -e^\top \Lambda e+e^\top rz-\theta r(r-1)\nonumber\\
    &&+\frac{c r^2}{2}\sum_{j=1}^ne_j^2\|\Delta_j \|^2\nonumber\\
    \!\!&\stackrel{\eqref{designk}}{=}&\!\!-\!\kappa\|z\|^2\!+\!\omega\!-\!k_1\|e\|^2\!-\!k_2r^2 \|e\|^2\!-\!\theta r(r\!-\!1)\!+\!e^\top rz\nonumber\\
    &\leq& -\kappa\|z\|^2+\omega-k_1\|e\|^2-k_2r^2\|e\|^2+ k_2r^2\|e\|^2\nonumber\\
    &&+\frac{1}{4k_2}\|z\|^2-\frac{\theta}{2}r^2+\frac{\theta}{2}\nonumber\\
     &=&-\left(\kappa-\frac{1}{4k_2}\right)\|z\|^2-k_1\|e\|^2-\frac{\ta}{2}r^2+\left(\frac{\ta}{2}+\omega\right)\nonumber\\
    &\leq&-\chi W+\left(\frac{\ta}{2}+\omega\right),\label{dotw}
\end{IEEEeqnarray}
where $\chi=\min\left\{2\kappa-\frac{1}{2k_2},2k_1,\theta\right\}$. From \eqref{dotw} we have
\begin{IEEEeqnarray}{rCl}
    r\leq \sqrt{2W(0)e^{-\chi t}+\frac{\theta+2\omega}{\chi}}=\varrho_r(t).\label{rboundr}
\end{IEEEeqnarray}
Recall that $e_d=rz$. From \eqref{ineqz} and \eqref{rboundr}, 
it is easy to conclude that $e_d$ is globally UUB. This completes the proof.
\hfill $\Box$
\end{proof}

\begin{remark}
    From \eqref{ineqz} and \eqref{rboundr}, one can see that $\lim_{t\to\infty} \|e_d(t)\|\leq \sqrt{\frac{\omega(\theta+2\omega)}{\kappa\chi}}$, implying that the ultimate bound of $e_d$ can be made arbitrarily small by appropriately choosing the tuning parameters $\gamma,\theta,c,k_1,k_2$ in the IIDOB design. In practice, the selection of these parameters should reflect a trade-off between reducing the ultimate disturbance estimation error and achieving a desired transient performance.
\end{remark}
\begin{remark}
From \eqref{dob:beta} one can see that $\beta$ is obtained via calculating an (indefinite) integral, whose explicit form is hard to obtain in general. In practice, a numerical integration can be adopted to compute $\beta$. Moreover, since $\psi\in C^1$, $\frac{\pa\be}{\pa\hat x}$ can be computed using the Leibniz integral rule \cite{folland1999real} as $\left(\frac{\pa\be}{\pa\hat x}\right)_{ij}=\int_0^{x_i}\frac{\pa}{\pa \hat x_j}\psi(\hat x_1,\cdots,\hat x_{i-1},\tau,\hat x_{i+1},\cdots,\hat x_n)\di \tau$,
where $\left(\frac{\pa\be}{\pa\hat x}\right)_{ij}$ denotes the $ij$-th entry of $\frac{\pa\be}{\pa\hat x}$.
\end{remark}

\begin{remark}
The IIDOB design method proposed in this work can be seen as an extension of the I\&I estimator design approach given in \cite{sonneveldt2010immersion,karagiannis2008observer} since unknown time-varying  functions, instead of unknown constant parameters, are considered in this work. Note that the definition of the scaled estimation error $z$ and the design of $\beta,\xi$ are different from their counterparts in \cite{sonneveldt2010immersion,karagiannis2008observer}.
\end{remark}

Before the end of this section, we show the design of an IIDOB-based tracking controller, which could be used as a nominal controller in the IIDOB-CBF-QP in Section \ref{sec:dobcbf}. 
Note that $\dot{\hat{d}}$ can be expressed as
\begin{equation}\label{dothatd}
\dot{\hat{d}}=\dot \xi+\frac{\pa\be}{\pa x}\dot x+\frac{\pa\be}{\pa\hat x}\dot{\hat{x}}+\frac{\pa\be}{\pa u}v\stackrel{\eqref{dob}}{=}-r\frac{\pa\beta}{\pa x}z.
\end{equation} The following proposition presents an IIDOB-based tracking control law provided the right inverse of $g$ exists.

\begin{proposition}\label{theorem:tracking}
Consider system \eqref{eqnsys} and suppose that all conditions of Theorem \ref{theorem:dob} hold such that the IIDOB shown in \eqref{dob} exists. Suppose that $\ka$ defined in \eqref{eqkappa} is greater than $1$, and the right inverse of $g$ exists for any $x\in \R^n$. Given a reference trajectory $x_d(t)$ where $x_d(t)$ and $\dot x_d(t)$ are bounded, $\forall t\geq 0$, if the control law is designed as 
\begin{IEEEeqnarray}{rCl}
\IEEEyesnumber\label{trackingcontrol}
\IEEEyessubnumber\label{trackingcontrol:ud}
    u_d&=&-g^{\dagger}\left(f+\left(\alpha_1+\fot r^2\right) e_x+\hat d-\dot x_d\right),\\
\IEEEyessubnumber \label{trackingcontrol:v}
    v&=&-\alpha_2e_u+\mathcal{G}_1-g^\top e_x-\frac{\|\mathcal{G}_2\|^2}{2}e_u,\label{vfinal}
\end{IEEEeqnarray}
where $e_x= x-x_d$, $e_u=u-u_d$, $g^{\dagger}$ is the right inverse of $g$, $\mathcal{G}_1=\frac{\pa u_d}{\pa t}+\frac{\pa u_d}{\pa x}(f+gu+\hat d)+\frac{\pa u_d}{\pa r}\dot r$, $\alpha_1,\alpha_2>0$ are positive constants, and $\mathcal{G}_2=r\left(\frac{\pa u_d}{\pa\hat d}\frac{\pa\beta}{\pa x}+\frac{\pa u_d}{\pa x}\right)$,
then the tracking error $e_x$ is globally UUB.
\end{proposition}
\begin{proof}
Define $V_1=\fot e_x^\top e_x+\fot z^\top z$ as a candidate Lyapunov function where $z$ is defined in \eqref{eqz}. Note that 
$\dot V_1\stackrel{\eqref{dotv0}}{\leq} e_x^\top (f+gu+pw-\dot x_d)-\kappa\|z\|^2+\omega
    \stackrel{\eqref{trackingcontrol:ud}}{\leq} -\left(\alpha_1+\fot r^2\right)\|e_x\|^2-re_x^\top z+e_x^\top g e_u-\kappa\|z\|^2+\omega\leq -\alpha_1\|e_x\|^2-\left(\kappa-\fot\right)\|z\|^2+e_x^\top ge_u+\omega$,
where $\omega$ is defined in \eqref{eqomega}. Since $u_d$ is a function of $x,r,\hat d$ and $t$, its derivative is $\dot u_d
    = \frac{\pa u_d}{\pa t}+\frac{\pa u_d}{\pa x}(f+gu+\hat d)+\frac{\pa u_d}{\pa r}\dot r-r\left(\frac{\pa u_d}{\pa\hat d}\frac{\pa\beta}{\pa x}+\frac{\pa u_d}{\pa x}\right)z=\mathcal{G}_1-\mathcal{G}_2z$.
Then, we define $V_2=V_1+\fot e_u^\top e_u$ as an augmented Lyapunov function candidate,  
whose derivative satisfies $\dot V_2\stackrel{\eqref{dothatd}}{\leq} \dot V_1+e_u^\top(v-\mathcal{G}_1+\mathcal{G}_2z)\leq -\alpha_1\|e_x\|^2-\left(\kappa-1\right)\|z\|^2+ e_u^\top g^\top e_x+e_u^\top(v-\mathcal{G}_1)+\frac{\|e_u\|^2\|\mathcal{G}_2\|^2}{2}+\omega\stackrel{\eqref{trackingcontrol:v}}{\leq}-\alpha_1\|e_x\|^2-(\kappa-1)\|z\|^2-\alpha_2\|e_u\|^2+\omega$.
Thus, $\dot V_2\leq -\vartheta V_2+\omega$,  
where $\vartheta=\min\{2\alpha_1,2\kappa-2,2\alpha_2\}$. Hence, one can see that $\|e_x\|\leq \sqrt{2V_2(0)e^{-\vartheta t}+\frac{2\om}{\vartheta}}$, indicating $e_x$ is globally UUB. This completes the proof.
 \hfill $\Box$
\end{proof}

When $g$ has no full row rank, an IIDOB-based tracking controller similar to Proposition \ref{theorem:tracking} can still be designed by following the backstepping technique \cite{KKK95}, provided some control Lyapunov function conditions hold. The details are omitted due to space limitation.

\begin{remark}\label{remark:surface}
The dynamic surface control \cite{swaroop2000dynamic} or command filter \cite{farrell2009command} technique can be adopted to bypass the tedious calculation of partial derivatives of $u_d$. For example, the idea of the dynamic surface control is to let $u_d$ defined in \eqref{trackingcontrol:ud} pass a low-pass filter
$\epsilon \dot u_d^f = -u_d^f+u_d$,
where $u_d^f$ is the filtered signal and $\ep$ is a small time constant. Then, one can replace $u_d$ with $u_d^f$ and use $\dot u_d^f$ directly in the design of $v$, instead of computing partial derivatives of $u_d$. 
\end{remark}

\section{IIDOB-CBF-QP-based Safe Controller}
\label{sec:dobcbf}
In this section, we will present an IIDOB-CBF-QP-based safe control design method to solve Problem \ref{prob2}. 

We will design a safe controller $v$ based on the augmented system shown in \eqref{eqnsys} that is used for the IIDOB design in the preceding section. Two issues need to be addressed in this design: (i) The time derivative of $\hat d$ is indispensable in control design and it depends on $z$, which is unknown since $z$ relies on $w$, as shown in \eqref{dothatd}; however, considering the worst-case of $\dot{\hat{d}}$ may lead to unnecessary conservatism. (ii) The minimum DRD of system \eqref{eqnsys} is lower than its minimum IRD (i.e., $d$ appears prior to $v$ when one differentiates $h$), which makes the direct decoupling of the disturbance from the system difficult even if the disturbance is precisely estimated \cite{yang2013static}.

We address the first challenge by designing a filter to obtain an alternative disturbance estimation signal whose derivative is known. 
Specifically, given an IIDOB shown in \eqref{dob}, we design the following  filter: 
\begin{equation}
    \dot{\hat{d}}_f=-\left(T_1+T_2r^2\left\|\frac{\pa \beta}{\pa x}\right\|^2\right)(\hat d_f-\hat d),\label{commandfilter1}
\end{equation}
where $\hat d_f$ denotes the filtered disturbance estimation with $\hat d_f(0)=\hat d(0)$, $r$ is governed by \eqref{dob:r}, $\beta$ is given in \eqref{dob:beta}, and $T_1,T_2>0$ are tuning parameters. 
The filter shown in \eqref{commandfilter1} is a modified low-pass filter by adding an additional  term $-T_2r^2\left\|\frac{\pa\beta}{\pa x}\right\|^2 (\hat d_f-\hat d)$ whose usefulness will be clear from the proof of Lemma 1.
From \eqref{commandfilter1} one can see that the derivative of $\hat d_f$ is completely known. 
Define the filtering error $e_f$ as 
\begin{equation}
    e_f=\hat d_f-\hat d.\label{filteringerror}
\end{equation}
The following result shows that $\hat d_f$ is close to $\hat d$ in the sense that $e_f$ is bounded by a known time-varying function whose ultimate bound can be arbitrarily small by choosing appropriate parameters.
\begin{lemma}\label{lemma:filteriidob}
Consider the augmented system \eqref{eqnsys}, the IIDOB as shown in \eqref{dob}, and the filter given in \eqref{commandfilter1}. If Assumption \ref{assmp:d} holds and $T_2>\frac{1}{4\ka}$, where $\ka$ is defined in \eqref{eqkappa}, then the filtering error $e_f$ satisfies
\begin{equation}
\|e_f(t)\|\leq \sqrt{\left(\|z(0)\|^2-\frac{2\om}{\zeta}\right)e^{-\zeta t}+\frac{2\om}{\zeta}}= \varrho_f(t)\label{efboundlemma1}
\end{equation}
for any $t\geq 0$, where $\zeta=\min\{2T_1,2\kappa-\frac{1}{2T_2}\}$ and $\om$ is defined in \eqref{eqomega}.
\end{lemma}
\begin{proof}
Substituting \eqref{dothatd} into
\eqref{commandfilter1} gives
\begin{equation}
    \dot e_f=-T_1e_f-T_2r^2\left\|\frac{\pa \beta}{\pa x}\right\|^2 e_f+r\frac{\pa\beta}{\pa x}z.\label{ezt}
\end{equation}
Construct a candidate Lyapunov function $V_f$ as 
\begin{equation}
    V_f=\fot e_f^\top e_f+\fot z^\top z,\label{eqnvf}
\end{equation}
whose derivative satisfies
\begin{IEEEeqnarray}{rCl}
\hspace{-0mm}\dot V_f\!\!&\stackrel{\eqref{dotv0}}{\leq}&\!\! -\!T_1\|e_f\|^2\!-\!T_2r^2\left\|\frac{\pa \beta}{\pa x}\right\|^2\!\|e_f\|^2\!+\!re_f^\top \frac{\pa\beta}{\pa x}z\!-\!\kappa\|z\|^2\!+\!\omega\nonumber\\
\!\!&\leq&\!\!\! -\!T_1\|e_f\|^2\!\!-\!\!T_2r^2\!\left\|\!\frac{\pa \beta}{\pa x}\!\right\|^2\!\!\!\!\|e_f\|^2\!\!+\!r\!\left\|\! \frac{\pa\beta}{\pa x}\!\right\|\!\|e_f\|\|z\|\!-\!\kappa\|z\|^2\!+\!\omega \nonumber\\
&\leq& -T_1\|e_f\|^2-\left(\kappa-\frac{1}{4T_2}\right)\|z\|^2+\omega\nonumber\\
&\leq& -\zeta V_f+\omega,\label{dotvz0}
\end{IEEEeqnarray}
where the second inequality comes from Cauchy-Schwarz inequality.
Note that $V_f(0)=\frac{1}{2}\|z(0)\|^2$ because $\hat d_f(0)=\hat d(0)$. It is standard to obtain \eqref{efboundlemma1} from \eqref{dotvz0}.\hfill $\Box$
\end{proof}

Next, we will present the safe controller design. Consider system \eqref{controlaffine} and the safe set $\mathcal{C}$ defined in \eqref{setc}. Suppose that $r_I= r_D=\iota\geq 1$ for system \eqref{controlaffine} with respect to $h$; clearly, because of \eqref{eqnsys2},  $r_D<r_I$ for the augmented system \eqref{eqnsys}. 
To address the second issue above, we define a set of functions $h_0,h_1,\dots,h_{\iota}$ as
\begin{IEEEeqnarray}{rCl}
\IEEEyesnumber\label{ecbf:total}
\IEEEyessubnumber\label{ecbf:1}
h_0(x)&=& h,\\
\IEEEyessubnumber\label{ecbf:2}
h_{i}(x)&=&\left(\frac{\di }{\di t}+\la_{i-1}\right)\circ h_{i-1}, \ i\in[\iota-1],\\
\IEEEyessubnumber
h_\iota(x,u,r,\hat d_f)&=&\frac{\pa h_{\iota-1}}{\pa x}(f+gu+\hat d_f)-\frac{(1+r^2)\left\|\frac{\pa h_{\iota-1}}{\pa x}\right\|^2}{2\tilde\rho(\zeta-\la_{\iota-1})}\nonumber\\
    &&-\tilde\rho\omega+\la_{\iota-1} h_{\iota-1},\label{definitionbarh:h1}
\end{IEEEeqnarray}
where
$\tilde\rho>0,\la_i>0 (i=0,1,\cdots,\iota-1)$ are tuning parameters, and $\omega>0$ is the constant defined in \eqref{eqomega}. 
Based on the IIDOB shown in \eqref{dob}, the filter given in \eqref{commandfilter1}, and the notations above, the following result provides a safe controller $v$ that ensures the forward invariance of $\mathcal{C}$ for system \eqref{controlaffine}.

\begin{theorem}\label{theorem:cbf}
Consider system \eqref{controlaffine} and the safe set $\mathcal{C}$ defined in \eqref{setc}. Suppose that all conditions of Theorem \ref{theorem:dob} hold such that the IIDOB shown in \eqref{dob} exists. Suppose that $r_I=r_D=\iota\geq 1$ for system \eqref{controlaffine} with respect to $h$, and there exist positive constants $\rho$, $\tilde\rho$, and $\la_i(i=0,1,\cdots,\iota)$ such that $\la_\iota<2\ka$, $\la_{\iota-1}<\zeta$, $h_i(x(0))>0$ for $i=0,\cdots,\iota-2$, $h_{\iota-1}(x(0))-\tilde\rho V_f(e_f(0),z(0))>0$, and $h_\iota(x(0),u(0),r(0),\hat d_f(0))-\frac{\rho}{2}\|z(0)\|^2>0$, where $\ka$, $\zeta$, and $V_f$ are defined in \eqref{eqkappa}, \eqref{efboundlemma1}, and \eqref{eqnvf}, respectively. If
$\sup_{v\in\R^m}[\psi_0+\psi_1v]\geq 0$
holds for any $u\in\R^m$, $r\in[1,\varrho_r]$, $\|\hat d\|\leq \|p\|\omega_0+\varrho_d$, $\|\hat d_f\|\leq \|p\|\om_0+\varrho_d+\varrho_f$, $\|e_f\|\leq \varrho_f$, $x\in\C$, and $t\geq 0$, where $\varrho_d=\varrho_z \varrho_r$ with $\varrho_z$ and $\varrho_r$ defined in \eqref{ineqz} and \eqref{rboundr}, respectively, $\varrho_f$ is defined in \eqref{efboundlemma1}, and
\begin{IEEEeqnarray}{rCl}
\IEEEyesnumber\label{cbfpsi01}
\psi_0 &=& \frac{\pa h_{\iota}}{\pa x}(f\!+\!gu\!+\!\hat d)+\frac{\pa h_{\iota}}{\pa r}\dot r\!+\!\frac{\pa h_{\iota}}{\pa \hat d_f}\dot{\hat{d}}_f-\frac{r^2\left\|\frac{\pa h_{\iota}}{\pa x}\right\|^2}{\rho(4\ka-2\la_\iota)}\nonumber\\
\IEEEyessubnumber\label{psi0}
&&-\rho\omega+\la_\iota h_{\iota},\\
\IEEEyessubnumber\label{psi1}
\psi_1&=& \frac{\pa h_{\iota}}{\pa u},
\end{IEEEeqnarray}
with $\dot r$, $\dot{\hat{d}}_f$, and $h_\iota$ defined in \eqref{dob:r}, \eqref{commandfilter1}, and \eqref{definitionbarh:h1}, respectively, 
then any Lipschitz controller $v\in K_{BF}=\{{\mathfrak v}\in\R^m:\psi_0+\psi_1 {\mathfrak v}\geq 0\}$
will ensure $h \geq 0$ for all $t\geq 0$.
\end{theorem}

\begin{proof}
Define  $H_1(x,u,r,\hat d_f,z)=h_\iota-\frac{\rho}{2}z^\top z$,
where $h_\iota$ is given in \eqref{definitionbarh:h1}. 
Since
\begin{IEEEeqnarray}{rCl}    \hspace{-2mm}\dot{H}_1\!&\stackrel{\eqref{dotv0}}{\geq}&\!\frac{\pa h_\iota}{\pa x}(f+gu+pw)+\frac{\pa h_{\iota}}{\pa u} v+\frac{\pa h_\iota}{\pa r}\dot r+\frac{\pa  h_\iota}{\pa \hat d_f}\dot{\hat{d}}_f\nonumber\\
&&+\rho\kappa\|z\|^2-\rho\omega\nonumber\\
    \!&=&\!\frac{\pa h_\iota}{\pa x}(f+gu+\hat d)+\frac{\pa h_{\iota}}{\pa u} v+\frac{\pa h_\iota}{\pa r}\dot r+\frac{1}{2}\rho\la_\iota\|z\|^2-\rho\omega\nonumber\\
    &&+\frac{\pa h_\iota}{\pa \hat d_f}\dot{\hat{d}}_f-r\frac{\pa h_\iota}{\pa x}z+\rho\left(\kappa-\frac{\la_\iota}{2}\right)\|z\|^2\nonumber\\
    \!&\geq&\!\frac{\pa h_\iota}{\pa x}(f+gu+\hat d)+\frac{\pa h_\iota}{\pa r}\dot r+\frac{\pa h_\iota}{\pa \hat d_f}\dot{\hat{d}}_f-\frac{r^2\left\|\frac{\pa h_\iota}{\pa x}\right\|^2}{\rho(4\ka-2\la_\iota)}\nonumber\\
    &&+\frac{1}{2}\rho\la_\iota\|z\|^2+\frac{\pa h_{\iota}}{\pa u} v-\rho\omega\nonumber\\
    \!&=&\! \psi_0+\psi_1v-\la_\iota H_1,
\end{IEEEeqnarray}
any $v\in K_{BF}$ will result in $\dot{H}_1\geq -\la_\iota\bar H_1$. Noting that $h_\iota(x(0),u(0),r(0),\hat d_f(0))-\frac{\rho}{2}\|z(0)\|^2>0\implies  H_1(x(0),u(0),r(0),\hat d_f(0),z(0))>0$,  
we can conclude that $H_1\geq 0$, which implies that $h_\iota\geq 0$, for any $t\geq 0$.

Define another function as $H_2(x,z,e_f)= h_{\iota-1}-\tilde\rho V_f$ 
where $V_f$ is given in \eqref{eqnvf}. Note that 
\begin{IEEEeqnarray}{rCl}
    &&\dot{H}_2+\la_{\iota-1} H_2\nonumber\\
    &\stackrel{\eqref{dotvz0}}{\geq}& \frac{\pa h_{\iota-1}}{\pa x}(f+g u+\hat d_f)-r\frac{\pa h_{\iota-1}}{\pa x}z-\frac{\pa h_{\iota-1}}{\pa x}e_f-\tilde\rho\omega\nonumber\\
    &&+\la_{\iota-1} h_{\iota-1}+\frac{(\zeta-\la_{\iota-1})\tilde\rho}{2}(\|z\|^2+\|e_f\|^2)\nonumber\\
    &\geq&\!\frac{\pa h_{\iota-1}}{\pa x}(f\!+\!g u\!+\!\hat d_f)\!-\!\frac{(1\!+\!r^2)\left\|\frac{\pa h_{\iota-1}}{\pa x}\right\|^2}{2\tilde\rho(\zeta-\la_{\iota-1})}\!-\!\tilde\rho\omega\!+\!\la_{\iota-1} h_{\iota-1}\nonumber\\
    &=& h_\iota\geq 0.
\end{IEEEeqnarray}
Since $h_{\iota-1}(x(0))-\tilde\rho V_f(e_f(0),z(0))>0\implies H_2(x(0),z(0),e_f(0))>0$, we have $H_2 \geq 0$, which implies that $h_{\iota-1}\geq 0$ for all $t\geq 0$. According to \eqref{ecbf:2}, one can conclude that $h_{i}\geq 0\implies h_{i-1}\geq 0$ for any $i\in[\iota-1]$ because $h_i(x(0))>0$, $i=0,1,\cdots,\iota-2$ \cite{nguyen2016exponential}. Therefore, one can see $h\geq 0$ for all $t\geq 0$. This completes the proof. \hfill $\Box$
\end{proof}

The safe controller $v$ proposed in Theorem \ref{theorem:cbf} can be obtained by solving the following convex IIDOB-CBF-QP:%
\begin{align}
\min_{v} \quad & \|v-v_{nom}\|^2\label{cbfQP}\\
\textrm{s.t.} \quad & \psi_0+\psi_1 v\geq 0, \nonumber
\end{align}
where $\psi_0,\psi_1$ are given in \eqref{cbfpsi01} and $v_{nom}$ is any given nominal control law that is potentially unsafe (e.g., the IIDOB-based tracking controller given in Proposition \ref{theorem:tracking}). Note that the IIDOB-CBF-QP \eqref{cbfQP} is convex and can be efficiently solved online. In fact, \eqref{cbfQP} has a  closed-form solution that can be expressed as \cite{ames2016control}:
\begin{equation*}
    v = \begin{cases}
    v_{nom},& \text{if}\  \psi_0+\psi_1 v_{nom} \geq 0,\\
    v_{nom} - \frac{\psi_0+\psi_1 v_{nom}}{\psi_1\psi_1^\top} \psi_1^\top, & \text{otherwise}.
    \end{cases}
\end{equation*}

\section{Simulation}
\label{sec:simulation}
In this section, two simulation examples are presented to demonstrate the effectiveness of the proposed method.

\begin{example}
Consider the following system:
\begin{equation}
    \hspace{-1mm}\begin{bmatrix}
    \dot x_1\\ \dot x_2
    \end{bmatrix}=
    \underbrace{\begin{bmatrix}
    x_2\\ x_1x_2   \end{bmatrix}}_{f(x)}+\underbrace{\begin{bmatrix}
    1&0\\0&1+\sin^2(x_1)
    \end{bmatrix}}_{g(x)}\begin{bmatrix}
    u_1\\u_2
\end{bmatrix}+\underbrace{\begin{bmatrix}x_1\\x_2   \end{bmatrix}}_{p(x)}w(t),\label{sim:numericalsystem}
\end{equation}
where $x=[x_1\ x_2]^\top$ is the state, $u=[u_1\ u_2]^\top$ is the control input, and $w$ is the disturbance. It is easy to verify that $f,g,p$ are smooth functions. 
In the simulation, we choose the disturbance as $w=5\sin(t)+2\cos(2t)+4\sin(3t)+3\cos(4t)$, which implies that Assumption \ref{assmp:d} holds with $\om_0=8,\om_1=26$. We select the initial condition as $x_1(0)=x_2(0)=-0.5$, and the reference trajectory of $x$ is $x_d(t)=[x_{1d}(t)\ x_{2d}(t)]^\top$, where $x_{1d}(t)=2\sin(t)$ and $x_{2d}(t)=2\cos(t)$.
 
We choose parameters $\ga=100$, $k_1=k_2=10$, $\theta=10$, $c=0.5$, and $\alpha_1=\alpha_2=50$ in Theorem \ref{theorem:dob} and Proposition \ref{theorem:tracking}. It is easy to verify that all conditions of Theorem \ref{theorem:dob} hold so that the IIDOB is designed as shown in \eqref{dob}. We design a tracking controller for \eqref{sim:numericalsystem} using the dynamic surface technique in \cite{swaroop2000dynamic} with the filter parameter $\ep=0.001$ (see Remark \ref{remark:surface}). As shown in Fig. \ref{fig:1}, the disturbance estimation of the proposed IIDOB is accurate and the tracking performance is satisfactory.

Next, we consider two safety sets $\mathcal{C}_1 = \{ x \in \R^2 : x_1+1 \geq 0\}$ and $\mathcal{C}_2 = \{ x \in \R^2 : 1-x_2 \geq 0\}$, which aim to keep $x_1\geq -1$ and $x_2\leq 1$. 
Define  $h_1=x_1+1$ and $h_2=1-x_2$. One can easily verify that the minimum IRD and the minimum DRD of system \eqref{sim:numericalsystem} with respect to $h_1,h_2$ are both 1, i.e., $r_I=r_D=1$.  
We choose parameters $\rho=\tilde\rho=1$, $\la_0=\la_1=50, T_1=50, T_2=1$ in Theorem \ref{theorem:cbf}, and let other parameters the same as above. We also choose the nominal controller as the tracking controller designed above. One can verify that all conditions of Theorem \ref{theorem:cbf} are satisfied, so that the IIDOB-CBF-QP-based controller \eqref{cbfQP} can ensure $h_1\geq 0$ and $h_2\geq 0$ for all $t\geq 0$. Indeed, as shown in Fig. \ref{fig:2}, trajectories of $x_1$ (or $x_2$) always remain within the safety set $\C_1$ (or $\C_2$). 

We also compare the tracking performance of the proposed controller \eqref{cbfQP} with the robust CBF approach proposed in \cite{jankovic2018robust}.
As shown in Fig. \ref{fig:2}, although the robust CBF controller can always ensure safety, its tracking performance of the reference trajectories inside the safe region is not as good as our proposed controller \eqref{cbfQP}.

\begin{figure}[!t]
\centering
\includegraphics[width=0.493\linewidth]{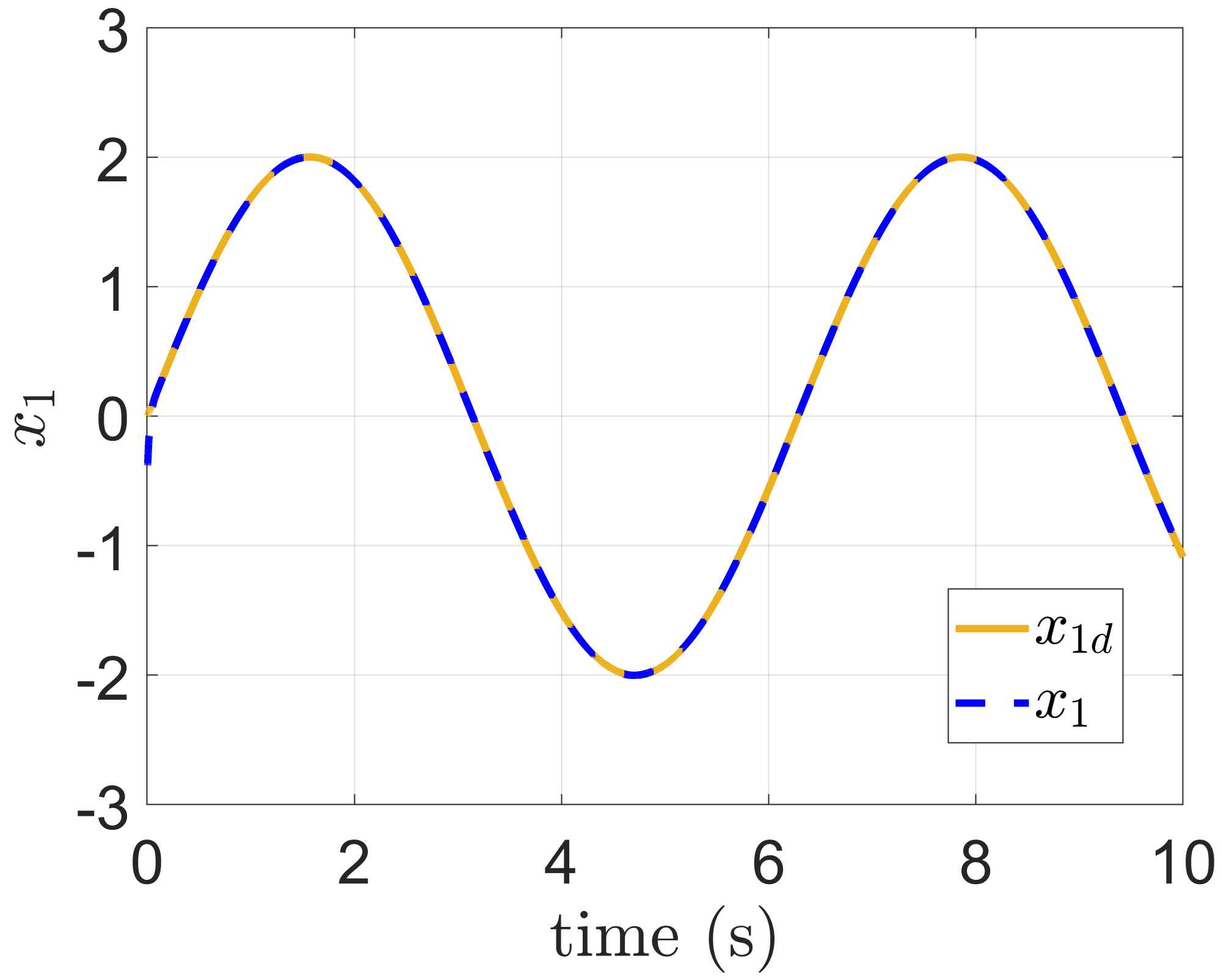}
\includegraphics[width=0.493\linewidth]{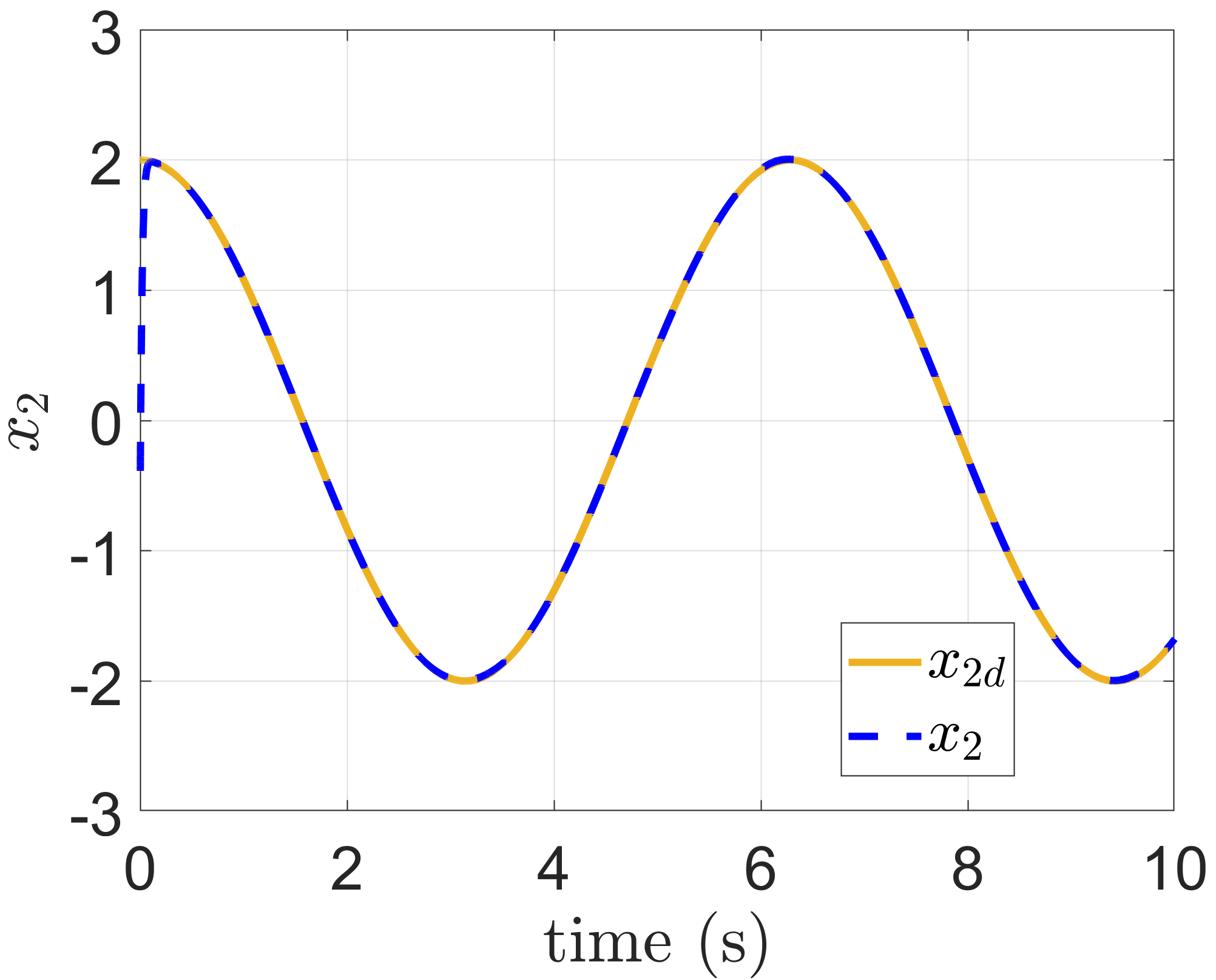}
\includegraphics[width=0.493\linewidth]{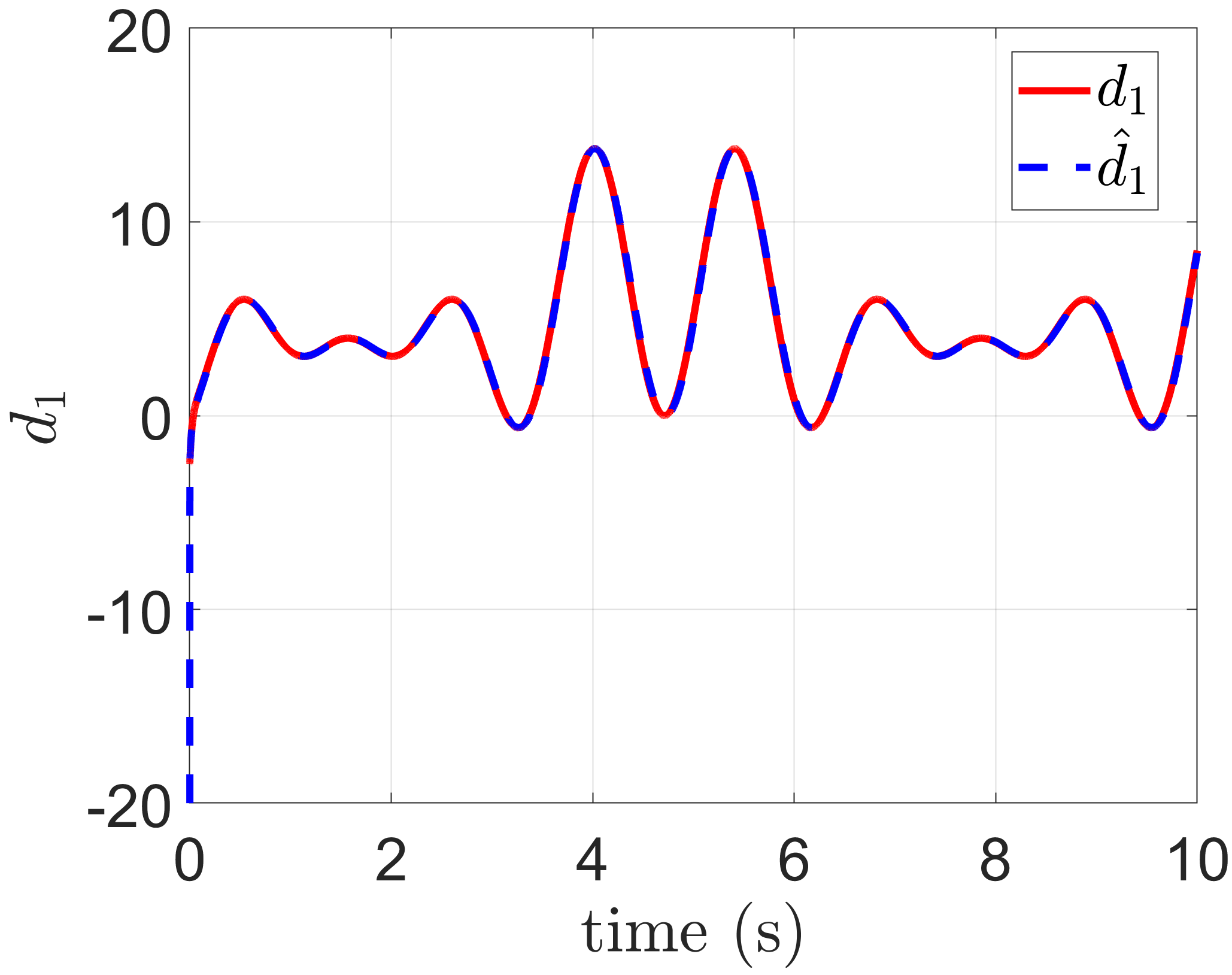}
\includegraphics[width=0.493\linewidth]{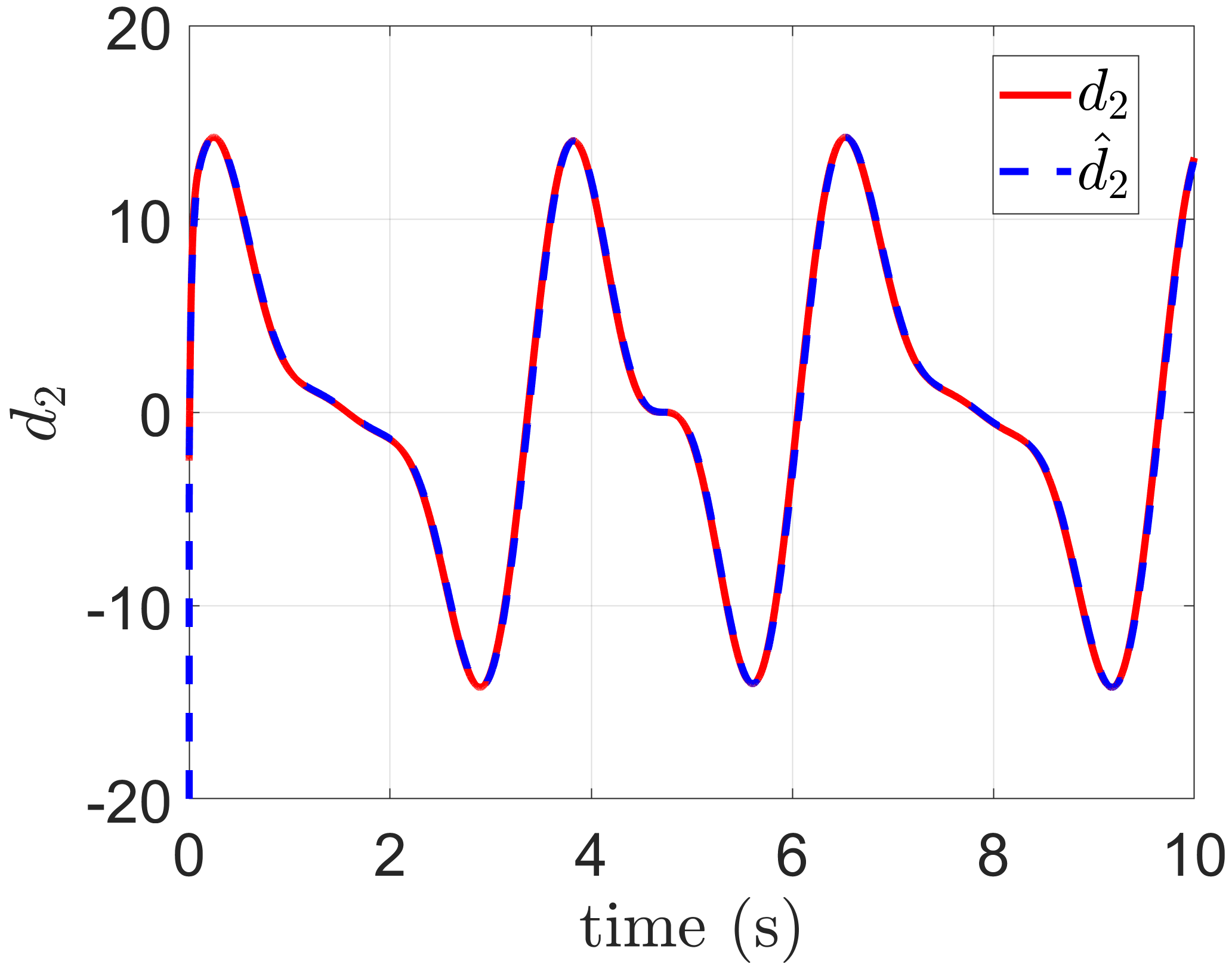}
\caption{Simulation results of the proposed IIDOB and IIDOB-based tracking controller for system \eqref{sim:numericalsystem}. (top) Trajectories of the states $x_1,x_2$ and the reference signals $x_{1d},x_{2d}$; (bottom) trajectories of the
total disturbances $d_1,d_2$ and the estimated total disturbances $\hat d_1,\hat d_2$.
}\label{fig:1}
\vskip -3mm
\end{figure}

\begin{figure}[!t]
\centering
\includegraphics[width=0.493\linewidth]{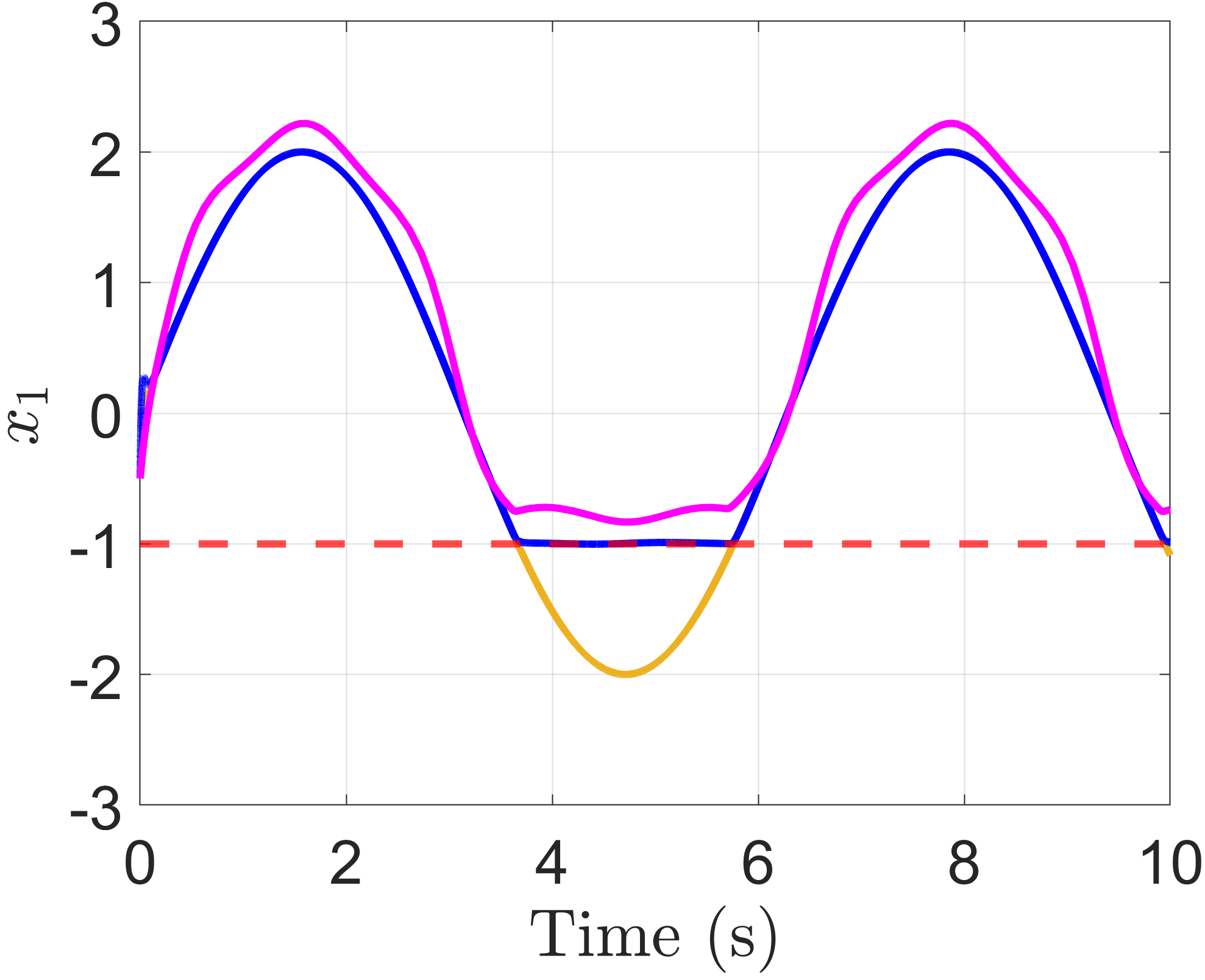}
\includegraphics[width=0.493\linewidth]{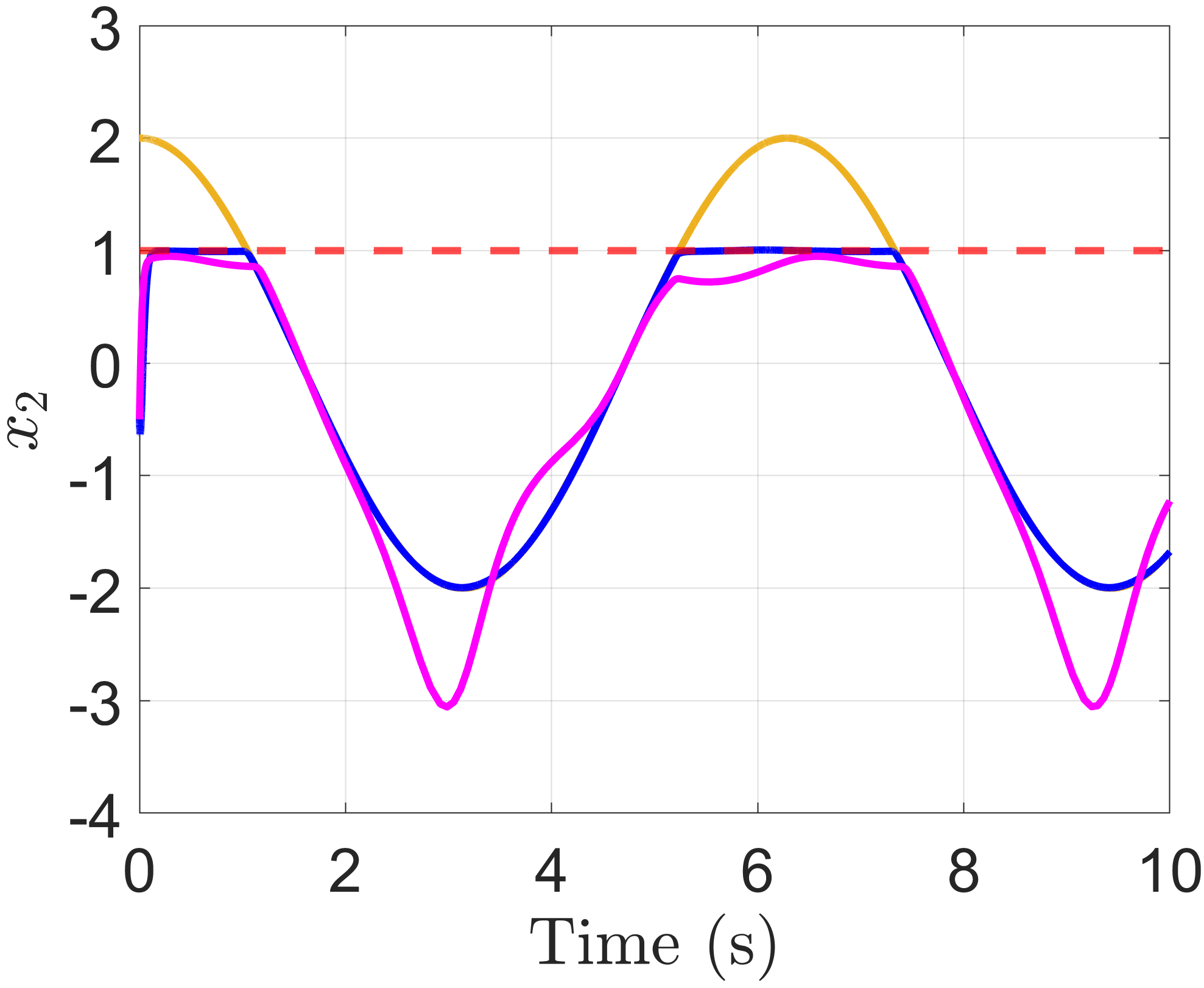}
\includegraphics[width=0.95\linewidth]{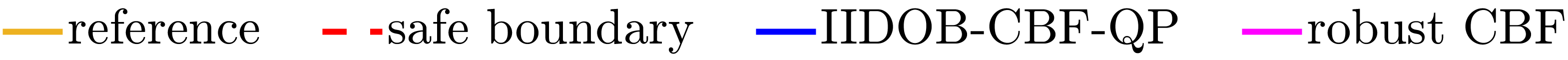}
\caption{Simulation results of the IIDOB-CBF-QP-based controller \eqref{cbfQP} and the robust CBF-based controller proposed in \cite{jankovic2018robust}. Both controllers can ensure safety, but controller \eqref{cbfQP} has better tracking performance inside the safe region.
}\label{fig:2}
\vskip -3mm
\end{figure}

\end{example}

\begin{example}
Consider a two-linked planar robot manipulator: 
\begin{equation}
    M(q)\ddot q+C(q,\dot q)\dot q+G(q)=\tau+J^\top(q)F_d(t),\label{eleqn}
\end{equation}
where $q=[q_1\ q_2]^\top$ is the joint angle, $\dot q=[\dot q_1\ \dot q_2]^\top$ is the joint angular velocity, $\tau\in\R^2$ is the control input, $F_d$ is the external disturbance satisfying Assumption \ref{assmp:d}, and $M\in\R^{2\times 2},C\in\R^{2\times 2},G\in\R^2$, and $J\in\R^{2\times 2}$ denote the inertia matrix, the Coriolis/centripetal matrix, the gravity term, and the Jacobian, respectively.

It can be seen \eqref{eleqn} can be expressed in the form of \eqref{controlaffine} with $x=[q_1\ q_2\ \dot q_1\ \dot q_2]^\top$, $f=[\dot q^\top \ -(C+G)^\top M^{-\top}]^\top$, $g=[0_{2\times 2}\ M^{-\top}]^\top$, and $p=[0_{2\times 2}\ JM^{-\top}]^\top$. It is obvious that $f,g,p$ are smooth functions. The expression of $M,C,G$ and physical parameters are chosen the same as those in \cite{sun2011neural}. 
Note that the Jacobian $J$ is singular when $q_2=0$ such that it is impossible to uniquely recover $F_d$  even if $\ddot q$ is available. The reference trajectory of $q$ is $q_d=[q_{1d}(t)\ q_{2d}(t)]^\top$, where $q_{1d}(t)=q_{2d}(t)=2\sin(t)$; the nominal controller is designed by following Proposition \ref{theorem:tracking}; the disturbance is selected as $F_d=[d_1\ d_2]^\top$ with $d_1=d_2=5\sin(t)+2\cos(2t)+4\sin(3t)+3\cos(4t)$ such that Assumption \ref{assmp:d} holds with $\om_0=11,\om_1=37$; four CBFs are selected as $h_1=q_1+1$, $h_2=-q_1+1.5$, $h_3=q_2+1.2$, and $h_4=-q_2+1$, which aim to ensure $-1\leq q_1\leq 1.5$ and $-1.2\leq q_2\leq 1$. It can be verified that the minimum IRD and the minimum DRD of system \eqref{eleqn} with respect to $h_i$, $i\in[4]$, are 2, i.e., $r_I=r_D=2$. 
We select the control parameters as $\ga=250, c=5, \ta=50, k_1=k_2=20, \alpha_1=\alpha_2=50, \rho=\tilde\rho=1, \la_0=25,\la_1=30,\la_2=50, T_1=250, T_2=1$ in Theorem \ref{theorem:dob} and \ref{theorem:cbf}. 

The simulation results are presented in Fig. \ref{fig:3}.  One can observe that the disturbance is precisely estimated by the proposed IIDOB \eqref{dob}, and the IIDOB-CBF-QP-based controller \eqref{cbfQP} can ensure the safety because the trajectories of $q_1,q_2$ remain within the safety sets.

\begin{figure}[!t]
\centering
\includegraphics[width=0.493\linewidth]{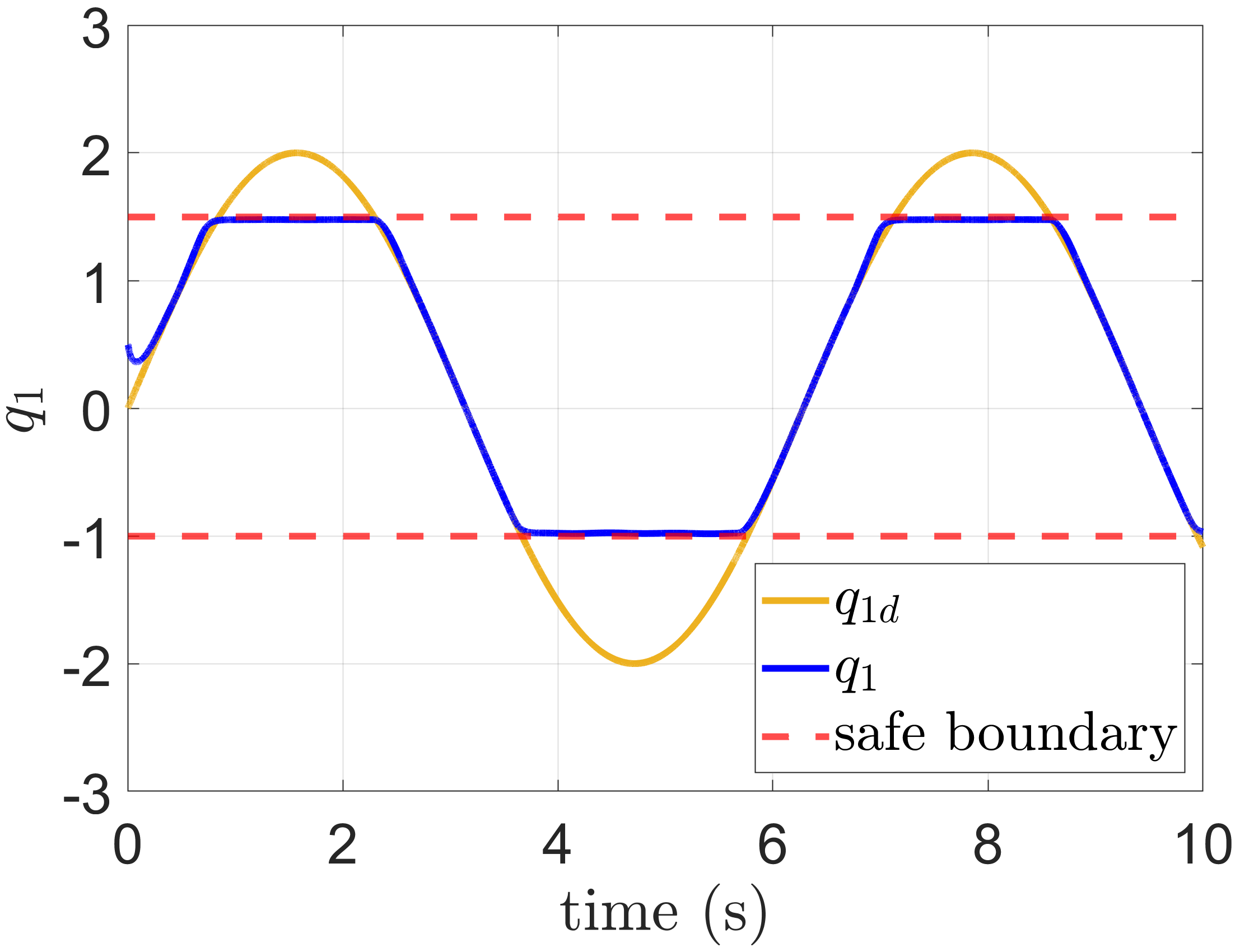}
\includegraphics[width=0.493\linewidth]{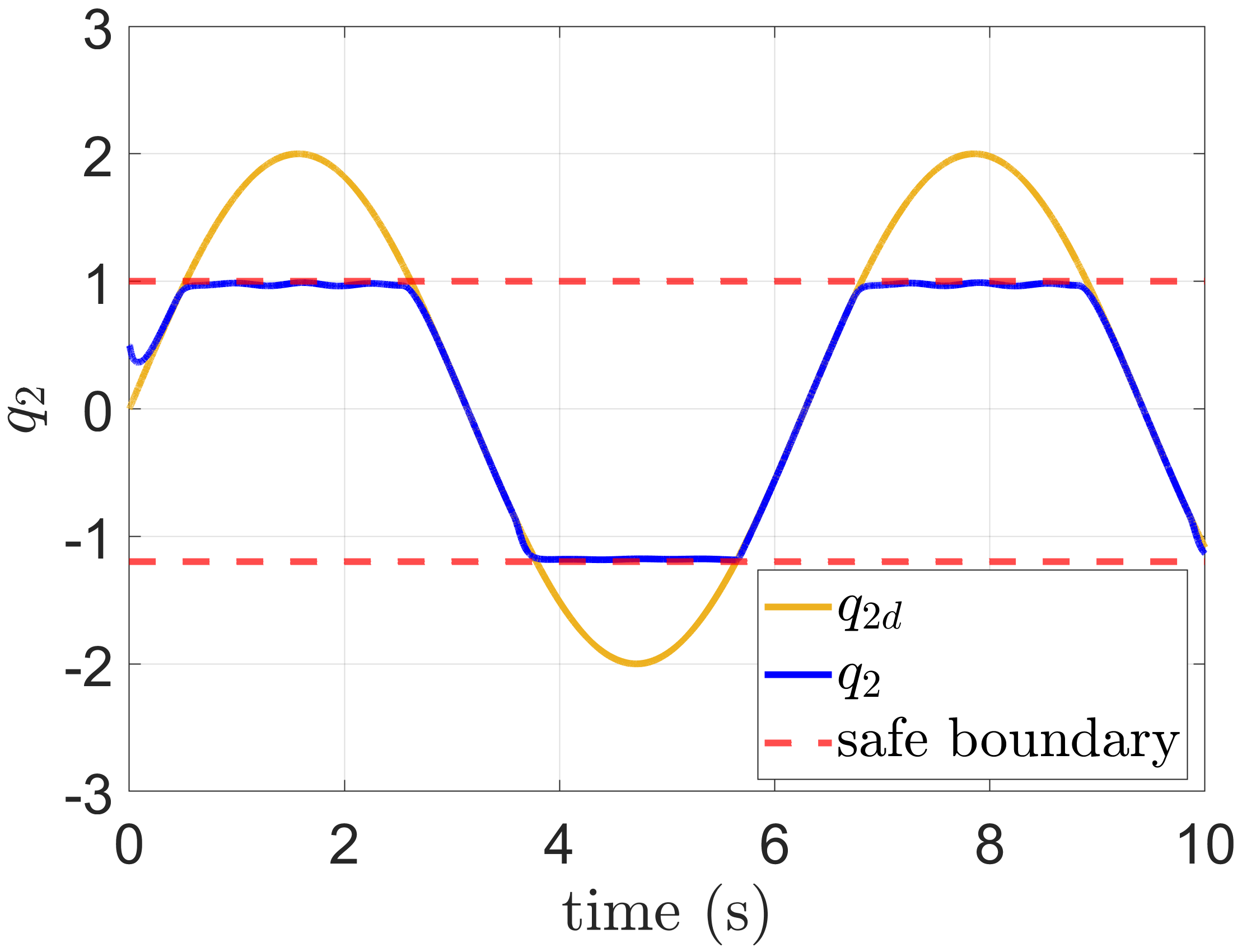}
\includegraphics[width=0.493\linewidth]{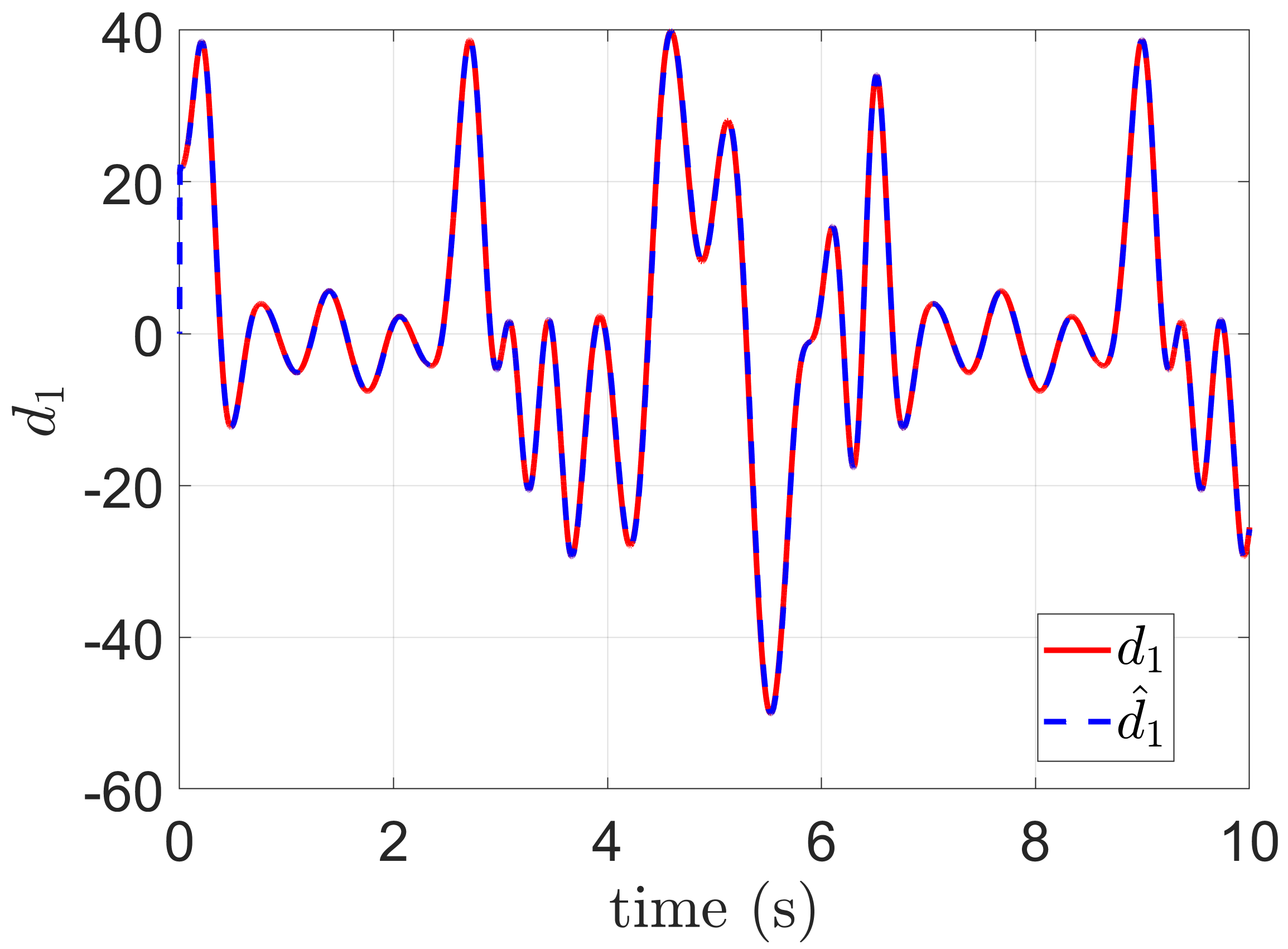}
\includegraphics[width=0.493\linewidth]{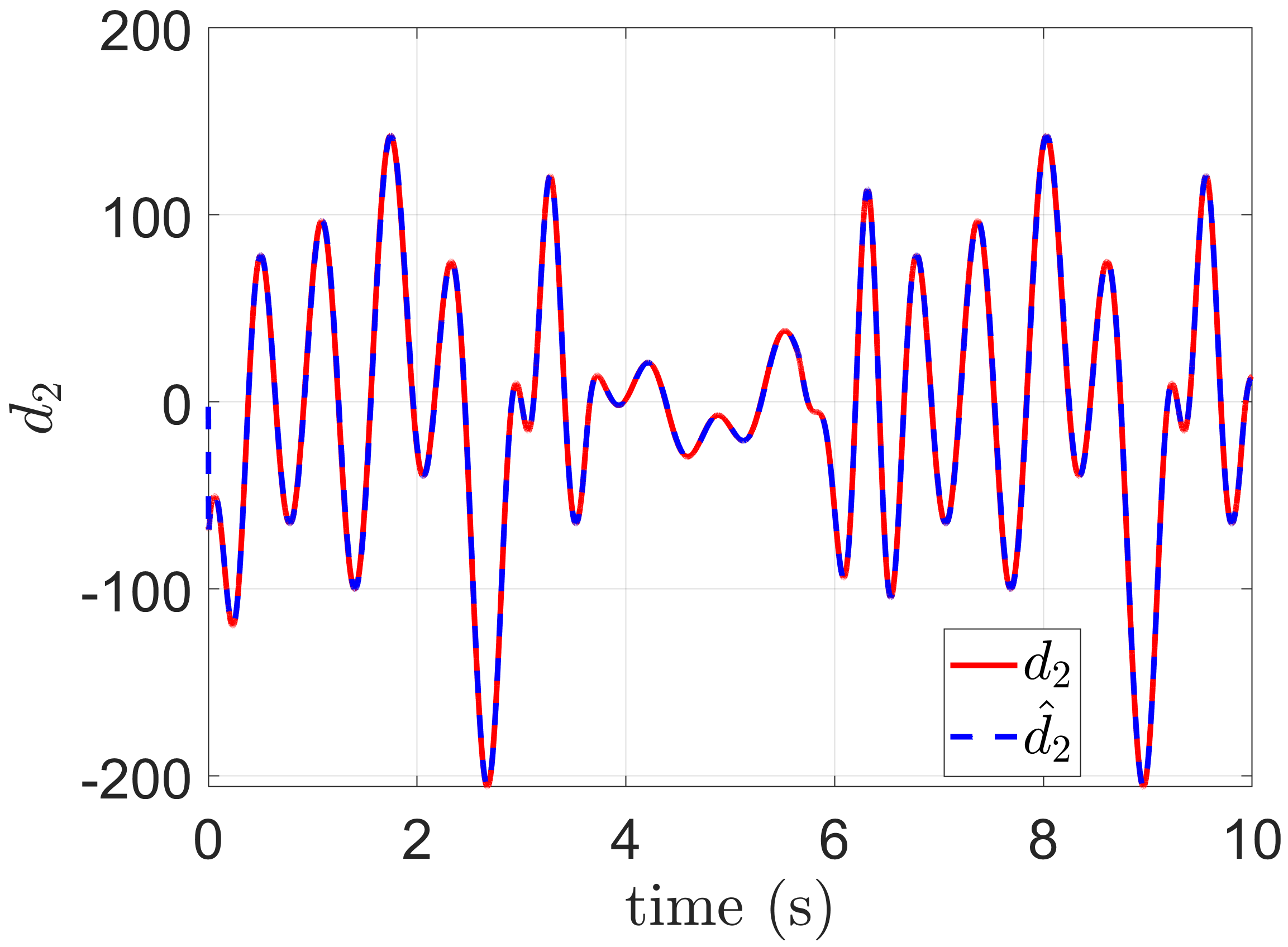}
\caption{Simulation results of the proposed IIDOB-CBF-QP-based safe control for the two-linked planer robot \eqref{eleqn}. One can see that the disturbance estimation is accurate and the safety is ensured. Moreover, the tracking performance of the nominal controller is well preserved inside the safe region. }\label{fig:3}
\end{figure}
\end{example}

\section{Conclusion}
\label{sec:conclusion}
This note introduced a systematic approach for designing IIDOB for general nonlinear control-affine systems without imposing restrictive assumptions employed by existing DOB design strategies.  
Based on that, a filter-based IIDOB-CBF-QP safe control design method was proposed. The numerical simulation results demonstrated the estimation accuracy achieved by the IIDOB and the superior performance of the proposed safe controller compared to the robust CBF-QP-based methods.
Future studies include extending the proposed method to hybrid systems and applying it to robotic systems such as bipedal robots subject to rigid impacts with the ground.

\bibliographystyle{IEEEtran}
\bibliography{DOBII.bib}

\end{document}